\documentclass[11pt]{article}

\usepackage{graphicx} 
\usepackage{dsfont}
\usepackage{epstopdf}
\usepackage{mathrsfs}
\usepackage{amsmath} 
\usepackage{amssymb}  
\usepackage{amsfonts,amsthm}
\usepackage{xcolor}
\usepackage{slashbox}
\usepackage{a4wide,epic}

\graphicspath{{./Figures/}}

\newtheorem{theorem}{Theorem}

\newcommand{\q}[1]{\vert #1 \rangle}
\newcommand{\qd}[1]{\langle #1 \vert}

\begin{document}

\title{Discrete-time reservoir engineering\\ with entangled bath and stabilizing squeezed states}

\author{Zibo Miao\footnote{Both authors are with QUANTIC lab, INRIA Paris, 2 rue Simone Iff, 75012 Paris, France. This research is carried out within Paris Sciences Lettres (PSL) Research University. AS is also with Electronics and Information Systems Department, Ghent University, Belgium. Corresponding author: zibo.miao@inria.fr}
~and Alain Sarlette$^\ast$}

\maketitle

\begin{abstract}
This theoretical proposal investigates how resonant interactions occurring when a harmonic oscillator is fed with a stream of entangled qubits allow us to stabilize squeezed states of the harmonic oscillator. We show that the properties of the squeezed state stabilized by this engineered reservoir, including the squeezing strength, can be tuned at will through the parameters of the ``input'' qubits, albeit in tradeoff with the convergence rate. We also discuss the influence of the type of entanglement in the input, from a pairwise case to a more widely distributed case. This paper can be read in two ways: either as a proposal to stabilize squeezed states, or as a step towards treating quantum systems with time-entangled reservoir inputs.
\end{abstract}





\section{Introduction}
\label{sec:Intro}
To stabilize a quantum system at a desired state plays a significant role in quantum information technology. However, the rather short dynamical time-scales of most quantum systems impose important limitations on the complexity of instantaneous output signal analysis and retroaction. An alternative control approach for quantum state stabilization, which bypasses a real-time analysis of output signals, is referred to as quantum reservoir engineering. A reservoir is designed such that, when acting alone, it drives the system initialized at arbitrary states, to a target state of the system coupled to the reservoir \cite{WHZ81,CMMRD01,ZWH03}. Reservoir engineering proposals have been developed both in continuous-time and discrete-time settings.
In continuous-time for example, a reservoir built from a combination of smartly chosen drives ensures the confinement of a system to a manifold of quantum states, spanned by the coherent superpositions of a few different coherent states \cite{LTP15}.
In the discrete-time case, a popular line of experiments uses a reservoir consisting of a stream of atoms consecutively interacting with a trapped field in order to stabilize the field's state. Such schemes have been devised to stabilize, without requiring any feedback action, a variety of states including squeezed states \cite{WDD08,SGDL07}, and so-called Schr\"odinger cat states \cite{SRBR11,SLBRR12}. In all these scenarios, the reservoirs are viewed as providing independent inputs at consecutive time-steps, progressively evacuating entropy to drive the system towards a target state.
In this paper, we consider a similar system setting to \cite{SRBR11,SLBRR12}, which is generally analogous to the Haroche experiment setting \cite{SDZ11} but with no measurement-based feedback involved. As shown in Fig. \ref{fig:GREng}, the reservoir consists of a sequence of input qubits, with the aim of manipulating and stabilizing the state of a harmonic oscillator indirectly via coupling with these qubits. The novelty of the present work is to investigate the effects of possibly \emph{entangling these consecutive inputs in time} before they interact with the oscillator. This idea draws upon \cite{GHK17}, where an abstract framework for bath-mediated controllability is formulated without assuming the qubits to be independent.

The motivation for this investigation is twofold.
First, by considering the bath qubits as a ``control input'', we want to establish the additional ``quantum power''  that is enabled by having \emph{entanglement in an input signal}, showing how it can be used to effectively improve the system performance. The result turns out to revolve around the stabilization of squeezed states of the harmonic oscillator mode. Second, we are seeking to develop general methods to analyze explicitly and efficiently the behavior of quantum systems under time-entangled inputs.

The relation between entangled qubits and field squeezing should not be too surprising. It is well-known that passing squeezed light through a 50/50 beamsplitter produces entangled outputs \cite{RMD89,TRK92}.
It has been extensively studied how coherent amplitude, relative phase and the squeezing of an optical field influence the quantum correlations (e.g. entanglement) between coupled qubits \cite{LAJP00,CP03,SDCZ01}. 
We here show that conversely, entangled input qubits stabilize a squeezed state of the oscillator. Since a squeezed state exhibits reduced noise below the vacuum level in one quadrature component of the field at the expense of amplified fluctuations in another component, it can be used to improve the signal-to-noise ratio, as a key step in the development of quantum communication \cite{VMC08,GDD13}.
Squeezing also notably enables quantum-enhanced metrology. For instance, sensitive interferometric and spectroscopic measurements are implemented by using the quadrature with reduced quantum noise of a propagating squeezed field state as a pointer for the measurement of weak signals, most notably in the last updates of gravitational wave detectors \cite{ATT14,MS90,CCM81,AAA13}. The present paper instead considers a trapped field, and \emph{stabilizes} it in a squeezed state. In such a setting, the reduced signal-to-noise ratio on one quadrature can be used to improve the sensitivity of this ``static probe'' to phenomena whose effect would be detected as a displacement along the squeezed direction -- e.g.~detecting electromagnetic driving fields acting on the probe, or improving readout schemes as in \cite{DBB15}.

In view of the importance of squeezed states, our observation of the possibility to stabilize strongly squeezed states thanks to an entangled qubits bath, might be worth noting for its own sake. Squeezed states on different physical platforms can be generated in so many different ways, that we have to refer the reader to dedicated work for a comprehensive review of both theoretical proposals and experimental realizations. The foremost domain is optics. It is important to distinguish the \emph{generation} of squeezed states from a given initial state, and the somewhat stronger achievement of \emph{stabilizing} squeezed states from any initial state. Protocols for the latter have recently been proposed in \cite{KMC13,DQB14,MM15J}, for a trapped electromagnetic field mode, and realized experimentally in \cite{PDB15}. 
Our paper likewise studies stabilization, but unlike in those papers, our goal here is to highlight the role played by entanglement in the bath states under a simple resonant interaction scheme.

From the viewpoint of investigating the possible consequences of time-entangled reservoir inputs, we thus establish the generation of squeezed states as a genuinely quantum effect. Indeed, since the direction of squeezing depends on the phase of the entangled qubits, replacing entanglement by a classical correlation would just lead to a thermally mixed steady state of the harmonic oscillator. We carry out an analysis in two cases. First, in a genuinely discrete-time setting (Section \ref{sec:pairs}), we analyze in detail the behavior of the oscillator interacting with pairwise entangled input qubits. We show how tuning their initial parameters allows, in principle, any squeezed state to be stabilized, but with the tradeoff that more steady-state squeezing implies slower convergence. We investigate the effect of perturbations on this scheme and give conditions under which a squeezing of the order of the benchmark 10dB squeezing factor \cite{VMC08} could be achieved. Second, in Section \ref{sec:interp} we investigate the case where the stream of qubits is weakly, locally entangled  and interacts weakly with the cavity. The methods and results established in this context are intended to suggest how to analytically approach the limit of continuous-time inputs.

\begin{figure}[!htp]
\centering
\includegraphics[scale=.6]{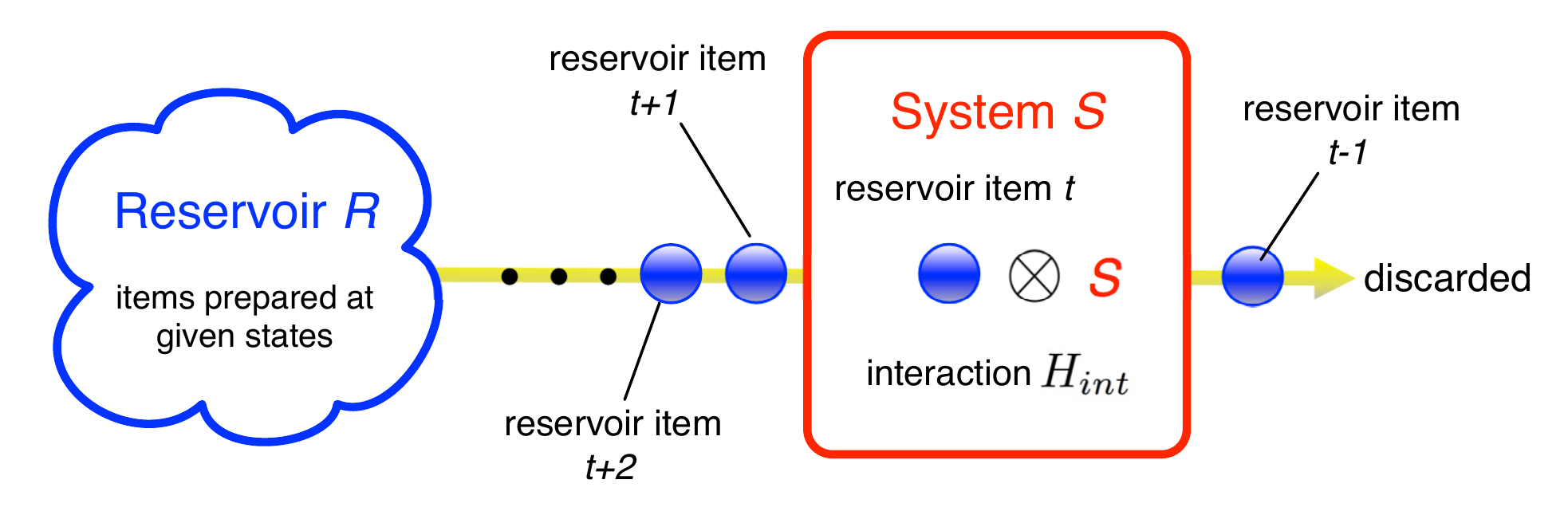}
\caption{(Colour online) Framework of quantum reservoir engineering. The aim is to stabilize the system S at a target state by coupling it to another quantum system R, referred to as the reservoir, which is viewed as a stream of input quantum states that are discarded after interaction. The novelty of the present work, unlike previous work where the reservoir is considered to consist of independent subsystems, is to establish the effect of entanglement among the reservoir items.}
\label{fig:GREng}
\end{figure}


\section{Discrete-time engineered reservoir setup}

As our general setup we consider, shown in Fig. \ref{fig:GREng}, a single harmonic oscillator mode of a cavity interacting resonantly with a sequence of qubits. Each qubit can be prepared in any desired state before entering the cavity, then interacts there with the stored oscillator mode, before exiting the cavity and essentially being lost while the next qubit starts its interaction. The stream of qubits then plays the role of a (discrete-time) engineered reservoir to control the oscillator's quantum state. This setup can be considered in the context of the ENS cavityQED type of experiments \cite{HR06B,DDC08,DAM93,BSM96,SDZ11}, 
where the harmonic oscillator mode corresponds to a standing wave of the electromagnetic field and the qubits states are Rydberg states of Rubidium atoms. We will sometimes use this more specific language for the sake of clarity, although the theoretical conclusions, of course, apply to other spin-spring systems, like in circuitQED \cite{WSB04,VKL13,RLSDM15} 
or possibly with trapped ions or mechanical oscillators \cite{PCZ96,MM07,BW08,HGSBBL14,PDB15,WLW15}.
It is now well-known that a short resonant interaction with a stream of \emph{independent} identical, weakly excited qubits can stabilize a coherent state inside the cavity, whereas more complicated interactions can stabilize e.g.~cat states of the field
\cite{SLBRR12}. The aim of the present paper is rather to investigate the (beneficial) effect of entanglement among the consecutive qubits. 

In more concrete terms, we consider the qubit-oscillator interaction as a fixed part of the experiment, which can later be fed with a ``controlled input'' stream of qubits. Each qubit will interact with the cavity for a fixed time $t_r$ according to the Jaynes-Cummings Hamiltonian
\begin{align}
\mathbf{H}_{JC} = i\frac{\Omega}{2}( | g \rangle \langle e |\mathbf{a}^{\dagger} - | e \rangle \langle g |\mathbf{a})
\end{align}
with $\Omega$ the effective qubit-oscillator coupling strength (Rabi oscillation frequency), $\mathbf{a}$ the oscillator mode's annihilation operator and $\q{g},\q{e}$ the qubit's ground and excited states. The unitary propagator describing one qubit-oscillator interaction is then
\begin{align}\label{eq:ResInt}
\mathbf{U}_r = | g \rangle \langle g |\cos\theta_{\mathbf{N}} +  | e \rangle \langle e |\cos\theta_{\mathbf{N+I}} - | e \rangle \langle g | \mathbf{a}\frac{\sin\theta_{\mathbf{N}}}{\sqrt{\mathbf{N}}} + | g \rangle \langle e |\frac{\sin\theta_{\mathbf{N}}}{\sqrt{\mathbf{N}}}\mathbf{a}^{\dagger}
\end{align}
where 
\begin{align}
\theta_{\mathbf{N}} = \theta\sqrt{\mathbf{N}}= \tfrac{1}{2}\Omega t_r \sum_n\sqrt{n}| n \rangle \langle n |, \;
\end{align}
with $\mathbf{N} = \mathbf{a}^\dagger \mathbf{a}$ the photon number operator, $\q{n}: n=0,1,2,...$ the Fock states of the harmonic oscillator mode, and $\mathbf{I}$ the identity operator. The joint state of the qubit-oscillator system after interaction is in general entangled. For a pure state $\q{\psi_c}$ of the oscillator before interaction, it can be described by
\begin{equation}\label{eq:factorit2}
\q{g'} \otimes \mathbf{M}_{g'} \q{\psi_c} + \q{e'} \otimes \mathbf{M}_{e'} \q{\psi_c}
\end{equation}
where $\q{g'},\q{e'}$ are any two orthonormal states of the qubit and $\mathbf{M}_{g'},\mathbf{M}_{e'}$ are operators on the harmonic oscillator satisfying $\mathbf{M}_{g'}^\dagger \mathbf{M}_{g'} + \mathbf{M}_{e'}^\dagger \mathbf{M}_{e'} = \mathbf{I}$. As the qubit state is dismissed after interaction, the latter operators express the (expected) evolution of the cavity by the Kraus map:
$$\rho_c \quad \mapsto \quad \mathbf{M}_{g'} \rho_c \mathbf{M}_{g'}^\dagger + \mathbf{M}_{e'} \rho_c \mathbf{M}_{e'}^\dagger \;, $$
where $\rho_c$ is the density operator describing the harmonic oscillator's state.

When this setup is fed with a stream of independent qubits labeled $t=1,2,3,...$, the cavity state thus undergoes the evolution
$$\rho_c(t+1) = \mathbf{M}_{g'}(t) \rho_c(t) \mathbf{M}_{g'}^\dagger(t) + \mathbf{M}_{e'}(t) \rho_c(t) \mathbf{M}_{e'}(t)^\dagger \; .$$
By varying in time the state of the ``input'' qubits before interaction, the operators $\mathbf{M}_{g'},\mathbf{M}_{e'}$ become time-dependent and propagate $\rho_c$. In the case of short resonant interaction, from the computations of \cite{SLBRR12}, at least for smooth variations in the input qubit states, this would allow us to stabilize the oscillator at a time-dependent coherent state, and no more. When the stream of qubits features classical correlations in time, it is clear that the oscillator will evolve as a classical (i.e.~non-coherent) superposition of coherent states, which is of no particular interest. When the stream of qubits features entanglement in time, the story becomes more interesting, as we show in the remainder of this paper.


\section{Resonant interaction with entangled qubit pairs}\label{sec:pairs}

We start with the situation where the stream of qubits is pairwise entangled. That is, the input to the cavity in fact consists of a stream of independent qubit pairs initialized in the general state
\begin{align}
\q{\psi_{q^2}} = \beta_{gg}|gg\rangle + \beta_{ge}|ge\rangle + \beta_{eg}|eg\rangle + \beta_{ee}|ee\rangle,
\label{eq:qubitis}
\end{align}
with $\beta_{gg},\beta_{ge},\beta_{eg},\beta_{ee}$ taking possibly time-dependent values in $\mathbb{C}$. Such a state can be prepared with a (partial) CNOT gate between the two qubits, possibly mediated through an auxiliary system. In the ENS cavityQED experiment, for instance \cite{HR06B,DDC08,DAM93,BSM96,SDZ11}, one might use an auxiliary cavity to entangle the flying qubits. The two qubits constituting an entangled pair still interact sequentially with the oscillator according to \eqref{eq:ResInt}. To obtain a Markovian evolution for the oscillator state, we must consider the result of its interaction \emph{with the qubit pair}. 
The latter is thus considered as one effective auxiliary system, which undergoes two consecutive Hamiltonian interactions with the oscillator (first Hamiltonian coupling with the subspace corresponding to first qubit, then with the subspace corresponding to second qubit), and the corresponding propagator is obtained as an obvious extension of $\mathbf{U_r}$. The evolution of the oscillator state over these \emph{two} interactions can then be described by the Kraus map:
\begin{align}
\rho_c(t+1) &= \mathbf{M}_{g'g'} \rho_c(t) \mathbf{M}_{g'g'}^\dagger + \mathbf{M}_{g'e'} \rho_c(t) \mathbf{M}_{g'e'}^\dagger + \mathbf{M}_{e'g'} \rho_c(t) \mathbf{M}_{e'g'}^\dagger + \mathbf{M}_{e'e'} \rho_c(t) \mathbf{M}_{e'e'}^\dagger \; . 
\label{eq:rhodyns} \end{align}
In the basis $(\q{g'},\q{e'}) = (\q{g},\q{e})$ for both qubits, the associated operators write:
\begin{align}
M_{gg} = & \beta_{gg}\cos^{2}\theta_{\mathbf{N}}+\beta_{ge}\frac{\cos\theta_{\mathbf{N}}\sin\theta_{\mathbf{N}}\mathbf{a}^{\dagger}}{\sqrt{\mathbf{N}}} +\beta_{eg}\frac{\sin\theta_{\mathbf{N}}\mathbf{a}^{\dagger}\cos\theta_{\mathbf{N}}}{\sqrt{\mathbf{N}}} +\beta_{ee}\frac{\left(\sin\theta_{\mathbf{N}}\mathbf{a}^{\dagger}\right)^{2}}{\mathbf{N}},\label{eq:M2qubits}\\ \nonumber
M_{ge} = &-\beta_{gg}\frac{\cos\theta_{\mathbf{N}}\mathbf{a}\sin\theta_{\mathbf{N}}}{\sqrt{\mathbf{N}}} +\beta_{ge}\cos\theta_{\mathbf{N}}\cos\theta_{\mathbf{N+I}} -\beta_{eg}\sin^2\theta_{\mathbf{N}} + \beta_{ee}\frac{\sin\theta_{\mathbf{N}}\mathbf{a}^{\dagger}\cos\theta_{\mathbf{N+I}}}{\sqrt{\mathbf{N}}},\\ \nonumber
M_{eg} = &-\beta_{gg}\frac{\mathbf{a}\sin\theta_{\mathbf{N}}\cos\theta_{\mathbf{N}}}{\sqrt{\mathbf{N}}}
-\beta_{ge}\frac{\sin^2\theta_{\mathbf{N+I}}(\mathbf{N+I})}{\mathbf{N}}\\ \nonumber
&+\beta_{eg}\cos\theta_{\mathbf{N+I}}\cos\theta_{\mathbf{N}}+\beta_{ee}\frac{\cos\theta_{\mathbf{N+I}}\sin\theta_{\mathbf{N}}\mathbf{a}^{\dagger}}{\sqrt{\mathbf{N}}},\\ \nonumber
M_{ee} = &\beta_{gg}\frac{\left(\mathbf{a}\sin\theta_{\mathbf{N}}\right)^{2}}{\mathbf{N}}-\beta_{ge}\frac{\mathbf{a}\sin\theta_{\mathbf{N}}\cos\theta_{\mathbf{N+I}}}{\sqrt{\mathbf{N}}}-\beta_{eg}\frac{\cos\theta_{\mathbf{N+I}}\mathbf{a}\sin\theta_{\mathbf{N}}}{\sqrt{\mathbf{N}}}+\beta_{ee}\cos^{2}\theta_{\mathbf{N+I}}. \nonumber
\end{align}


\subsection{Approximate analysis}

Assuming $\theta \ll 1$, we expand this Kraus map to second order in $\theta$ and we observe that it appears to be the discretization of a Lindblad master equation with 3 dissipation channels:
\begin{align}
\tfrac{d}{d\tau}\rho_c(\tau) =-i\left[\mathbf{H},\rho_c(\tau)\right]+\sum_{j=1}^3\mathcal{D}\left(\mathbf{L}_{j}\right)\rho_c\left(\tau\right),
\label{eq:LindM}
\end{align}
with $\mathcal{D}(\mathbf{L})\rho_c=\mathbf{L}\rho_c\mathbf{L}^{\dagger}-\frac{1}{2}(\mathbf{L}^{\dagger}\mathbf{L}\rho_c+\rho_c \mathbf{L}^{\dagger}\mathbf{L})$, and operators
\begin{align}
\mathbf{H} &= -i\theta\left(\mathbf{Q}-\mathbf{Q}^{\dagger}\right),\nonumber\\
\mathbf{Q} &=\left[\beta_{gg}\left(\beta_{ge}^{\ast}+\beta_{eg}^{\ast}\right)+\beta_{ee}^{\ast}\left(\beta_{ge}+\beta_{eg}\right)\right]\mathbf{a},\nonumber\\
\mathbf{L}_1 &= \sqrt{2}\theta\left(\beta_{gg}\mathbf{a}  - \beta_{ee}\mathbf{a}^\dagger\right),\nonumber\\
\mathbf{L}_2 &= \theta\left(\beta_{ge}+\beta_{eg}\right)\mathbf{a},\nonumber\\
\mathbf{L}_3 &= \theta\left(\beta_{ge}+\beta_{eg}\right)\mathbf{a}^\dagger \; .
\label{eq:2appie}
\end{align}
The adimensional time $\tau$ corresponds to the duration of the interaction with one pair of qubits. The decoherence operators $\mathbf{L}_2$ and $\mathbf{L}_3$ describe a purely thermal bath at infinite temperature; this simply has the effect of stabilizing a high-energy thermal mixture of coherent states. However by taking $|\beta_{eg}|,|\beta_{ge}| \ll 1$, we can make the coupling to this thermal reservoir negligible leaving the dominant terms $\mathbf{L}_1$ and $\mathbf{H}$, which involve $\theta$ whereas the dissipation super-operator is in $\theta^2$. In this regime, the reservoir can be tuned to stabilize any minimum-uncertainty squeezed state
\begin{align*}
\q{\alpha,\zeta=r e^{i\phi_r}} & =\mathbf{D}\left(\alpha\right) \mathbf{S}\left(\zeta\right) \q{0} \quad \text{where}\\
& \mathbf{D}\left(\alpha\right)=\exp\left(\alpha \mathbf{a}^{\dagger}-\alpha^{\ast}\mathbf{a}\right) \; ,\\
& \mathbf{S}\left(\zeta\right)=\exp\left(\tfrac{1}{2}(\zeta^{\ast}\mathbf{a}^2-\zeta(\mathbf{a}^{\dagger})^{2})\right)
\end{align*}
are respectively the displacement operator by $\alpha = |\alpha|e^{i\phi_\alpha} \in \mathbb{C}$ and the squeezing of the vacuum by $|\zeta|=r$ in a direction $\phi_r$. That is, denoting $\mathbf{X_\phi} = \frac{\mathbf{a}e^{i\phi}+\mathbf{a}^\dagger e^{-i\phi}}{2}$ the oscillator quadrature in direction $\phi$ (we follow the convention from \cite{HR06B}), the corresponding variances satisfy
\begin{align*}
& (\Delta\mathbf{X}_{\frac{\phi_r}{2}}) (\Delta\mathbf{X}_{\frac{\phi_r+\pi}{2}}) = \frac{1}{4}, \\
& (\Delta\mathbf{X}_{\frac{\phi_r}{2}}) = \tfrac{1}{2}\, e^{-r}
\end{align*}
for $\q{\psi} = \mathbf{S}\left(r\right) \q{0}$. Thus such state saturates the Heisenberg uncertainty inequality, with less uncertainty on $\mathbf{X}_{\phi_r/2}$ than a classical-like state like the vacuum $\q{0}$.
\vspace{2mm}

\begin{theorem}
Consider the Lindblad master equation
\begin{align}
\tfrac{d}{dt}\rho_c(\tau) =-i\left[\mathbf{H},\rho_c(\tau)\right]+\mathcal{L}\left(\mathbf{L}_{1}\right)\rho_c\left(\tau\right)
\label{eq:4thm1}
\end{align}
which describes, according to approximations just discussed, the engineered reservoir obtained through resonant interaction of a harmonic oscillator with a stream of consecutive entangled qubit pairs initialized in the state \eqref{eq:qubitis} before interaction. This Lindblad master equation stabilizes the squeezed state $\q{\alpha,r e^{i\phi_r}}$ provided we initialize the qubit pairs as
\begin{align*}
& \beta_{gg} = \cos\epsilon\cos u \;\;,\;\; \beta_{ee} = e^{i\mu}\cos\epsilon\sin u \;\; ,\\
& \beta_{eg}=\beta_{ge}=e^{i\chi}\sin\epsilon /\sqrt{2} \; .
\end{align*}
with the parameters tuned as:
\begin{align}
\mu &= \phi_r,\label{eq:findmu} \\
\tan u &= -\tanh r\;, \;\;\; u\in (-\pi/4,\pi/4), \label{eq:findu}\\
\tan \chi &= \frac{\sin \phi_{\alpha}-\tanh r \sin \left(\phi_{\alpha} - \phi_r\right)}{\cos \phi_{\alpha} + \tanh r  \cos \left( \phi_{\alpha} - \phi_r\right)},\label{eq:findvt} \\
\epsilon &= \frac{\theta|\alpha|}{\sqrt{2}}\sqrt{\frac{\left(1+\sin(2u)\right)\left(1-\sin(2u)\right)}{1+\cos(\phi_{r} - 2\chi)\sin (2u)}}.\label{eq:findet}
\end{align}
The convergence rate towards $\q{\alpha,re^{i\phi_r}}$ is
\begin{align}
\kappa = 2\theta^2\frac{1-\tanh^2 |r|}{1+\tanh^2 |r|}.
\label{eq:findka}
\end{align}
\label{thm:con}
\end{theorem}
\begin{proof}
We define a new density operator
\begin{align}
\tilde{\rho}_c (t)= \mathbf{S}^\dagger(\zeta)\mathbf{D}^\dagger(\alpha)\rho_c(t) \mathbf{D}(\alpha)\mathbf{S}(\zeta) \;.
\label{eq:Rhotran}
\end{align}
Using the properties 
\begin{align*}
\mathbf{D}^\dagger(\alpha)\mathbf{a}\mathbf{D}(\alpha) &= \mathbf{a} + \alpha,\nonumber\\
\mathbf{S}^\dagger(re^{i\phi_r})\mathbf{a}\mathbf{S}(re^{i\phi_r}) &= \mathbf{a}\cosh r - e^{i\phi_r}\mathbf{a}^\dagger \sinh r \nonumber
\end{align*}
and a few straightforward computations, with the parameter tuning provided in the statement, we obtain
\begin{align}
\tfrac{d}{d\tau}\tilde{\rho}_c\left(\tau\right) = \kappa\; \mathcal{D}(\mathbf{a}) \tilde{\rho}_c(\tau).
\label{eq:LindMvac}
\end{align}
This equation stabilizes $\tilde{\rho}_c$ towards the vacuum state at a rate $\kappa$. The converse change of variables yields the result.

We further note that the suggested tuning for $\epsilon$ involves
$$ \sqrt{\frac{\left(1+\sin(2u)\right)\left(1-\sin(2u)\right)}{1+\cos(\phi_r - 2\chi)\sin (2u)}} \in (0,\sqrt{2})$$
i.e.~a bounded value for any bounded $\alpha$, squeezing rate $r$ and directions $\phi_\alpha,\phi_r$. The other parameters remain bounded and the fractions are never undefined provided $r$ remains finite i.e.~$|u| < \pi/4$. The convergence rate $\kappa$ is strictly positive and it approaches zero as $r$ approaches $\pm \infty$ i.e.~$|u|$ approaches $\pi/4$.
\end{proof}

Theorem \ref{thm:con} shows how any squeezed state can be stabilized by \eqref{eq:4thm1}. It also covers the convergence properties of this equation for all nontrivial cases. Indeed, when $\pi/4 <|u|<\pi/2$ for the qubit pairs, the cavity state will just drift away to unbounded energies, which is not of much interest in quantum technology. For $u = \pm \pi/4$, the operator $\mathbf{L}_1$ becomes proportional to a Hermitian measurement operator of  $\mathbf{X}_{\phi_r/2}$, i.e.~one quadrature of an electromagnetic field mode or a particular linear combination of position/momentum in a mechanical oscillation mode. Any eigenstates of this Hermitian operator are invariant states of the system. These are infinitely squeezed, but are stabilized infinitely slowly.

The convergence rate $\kappa$ in adimensional time $\tau$, expresses the evolution of $\rho_c$ as a function of the number of 2-qubit interactions. If those interactions are separated by zero delay, then it corresponds to units of $2t_r = \tfrac{4\theta}{\Omega}$, where $\Omega$ is the Rabi oscillation frequency describing the effective qubit-oscillator coupling strength. 

The result also emphasizes the role of entanglement, beyond simply classical correlations among the qubits, in order to stabilize a squeezed state. For example, for $\epsilon=0$, classical correlations would correspond to a superposition of states with various values of $\mu$. These would correspond to trying to stabilize squeezed states with various orientations $\phi_r$, whose superposition yields a thermally diffused steady state, with higher uncertainties in all quadratures than a coherent state. We have checked that this does indeed occur in corresponding simulations of our reservoir.


\subsection{Imperfection effects and Simulation results}

We have carried out simulations on a system, once under ideal conditions, and once when it contains typical imperfections.

For the simulations of an idealized experiment, we propagate the cavity state, which is initially prepared at the vacuum state $|0\rangle$, according to \eqref{eq:rhodyns},\eqref{eq:M2qubits}. The qubit-oscillator coupling is set to $\theta=\pi/20$. We have represented $\rho_c$ in the Fock basis, truncated at a high number of photons. Figure \ref{bigfig:ideal2} illustrates the tuning capabilities of the reservoir with pairwise entangled qubits by varying their parameters. We provide the steady state values of $\langle \mathbf{X}_\phi \rangle_{\rho_c}$ and $(\Delta\mathbf{X}_\phi)_{\rho_c}$ expected from Theorem 1, as well as those observed in simulations. In particular for $\Delta\mathbf{X}_\phi$ we give the extreme values, corresponding in theory to $\phi = \phi_r/2$ and $\phi=\phi_r/2+\pi/2$. The squeezing power in dB, as usually expressed in the literature \cite{VMC08}, corresponds to $r_{\text{eff}} \text{ (dB)} = w \, \min_{\phi} \ln\left(2\Delta\mathbf{X}_{\phi}\right) \, \text{dB}$ with $w = \frac{-20}{\ln\left(10\right)} \simeq -8.686$. All the simulation results shown feature a good agreement with the approximate theory. In particular, the state always converges very close to a minimum uncertainty state $(\Delta\mathbf{X}_{\phi_r/2})_{\rho_c} (\Delta\mathbf{X}_{\phi_r/2+\pi/2})_{\rho_c} = 0.250 \pm 0.0005$. For $\epsilon=0$ this is a consequence of the small value of $\theta$. But also for $\epsilon>0$ the approximation, dropping $\mathbf{L}_2$ and $\mathbf{L}_3$ from \eqref{eq:LindM}, appears to hold reasonably well, although a slight decrease in squeezing power is observed in simulation. 
A squeezing factor of 10 dB, as presented in the figure (top right), is still considered a significant benchmark \cite{VMC08} while the world record, by the group of Roman Schnabel \cite{VMD16}, is around 17 dB.

\begin{figure*}
\centering
\setlength{\unitlength}{1mm}
\begin{picture}(190,185)(10,0)


\put(30,155){\includegraphics[scale=.4, trim=45mm 8mm 10mm 8mm, clip=true]{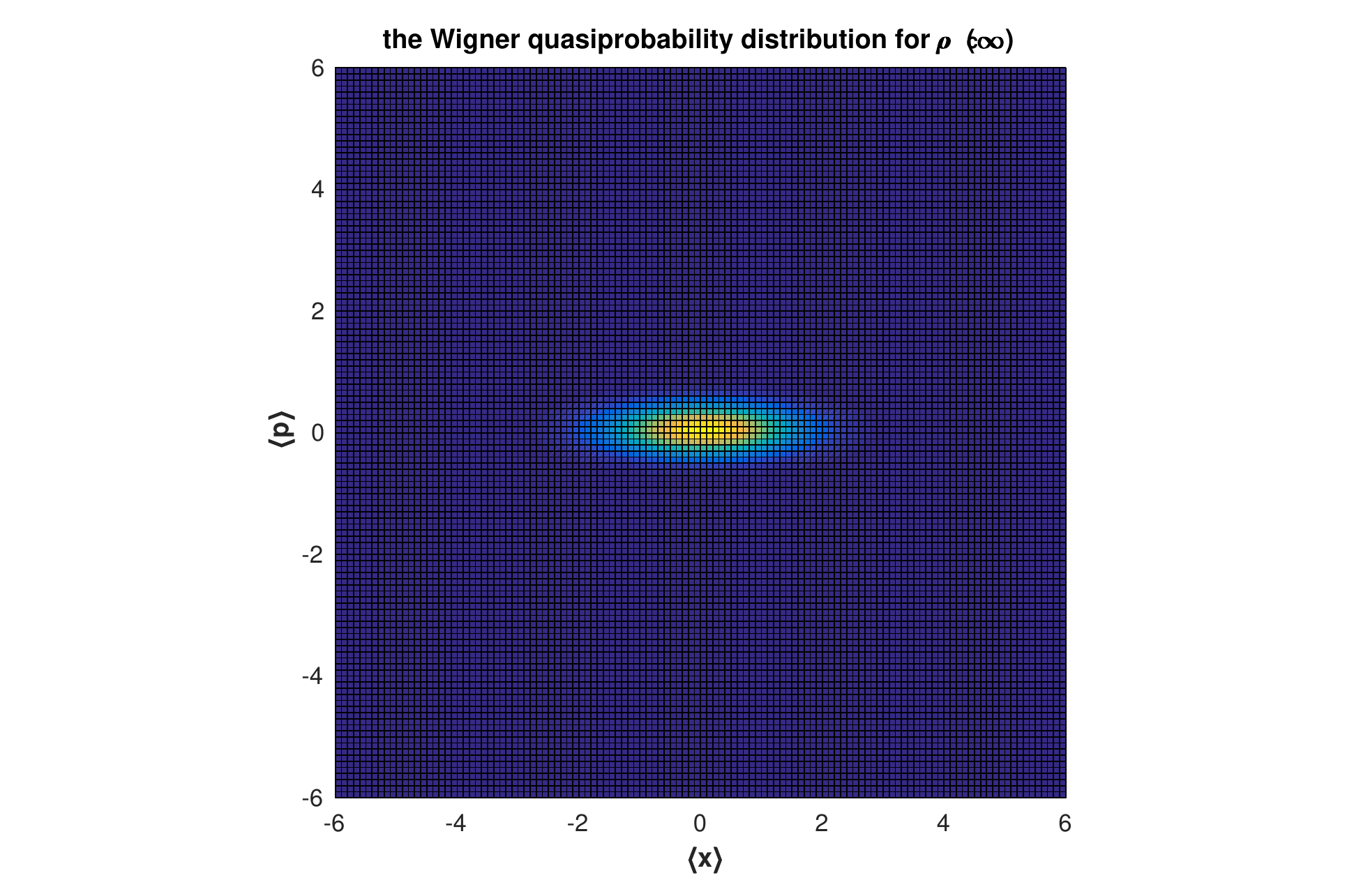}}
\put(0,177){
\tiny
\begin{tabular}{|c|}
\hline
$\q{\psi_{q^2}}$ setting \\ \hline $u=\dfrac{\pi}{6}\vphantom{\dfrac{A^K}{A_K}}$ \\ \hline $\mu=\pi$ \\ \hline $\epsilon=0$ \\ \hline $\chi=0$ \\
\hline
\end{tabular}}
\put(0,144){
\tiny
\begin{tabular}{|c|c|c|c|c|c|c|c|}
\hline 
 & $\left\langle \mathbf{X}_0 \right\rangle $ & $\left\langle \mathbf{X}_{\frac{\pi}{2}}\right\rangle $ & $\phi_r$ & $\Delta\mathbf{X}_{\frac{\phi_r}{2}}$ & $\Delta\mathbf{X}_{\frac{\phi_r+\pi}{2}}$ & $\kappa$ & $r_{\text{eff}}$(dB)\\
\hline 
\hline 
the. & 0 & 0 & $0$ & 0.966 & 0.259 & 0.025 & 5.7 \\
\hline 
sim. & 0 & 0 & $0$ & 0.970 & 0.258 & 
& 5.75 \\
\hline 
\end{tabular}}
\put(77,156){\scriptsize Re$\alpha$}
\put(24,196){\scriptsize Im$\alpha$}


\put(122,155){\includegraphics[scale=.4, trim=45mm 8mm 10mm 8mm, clip=true]{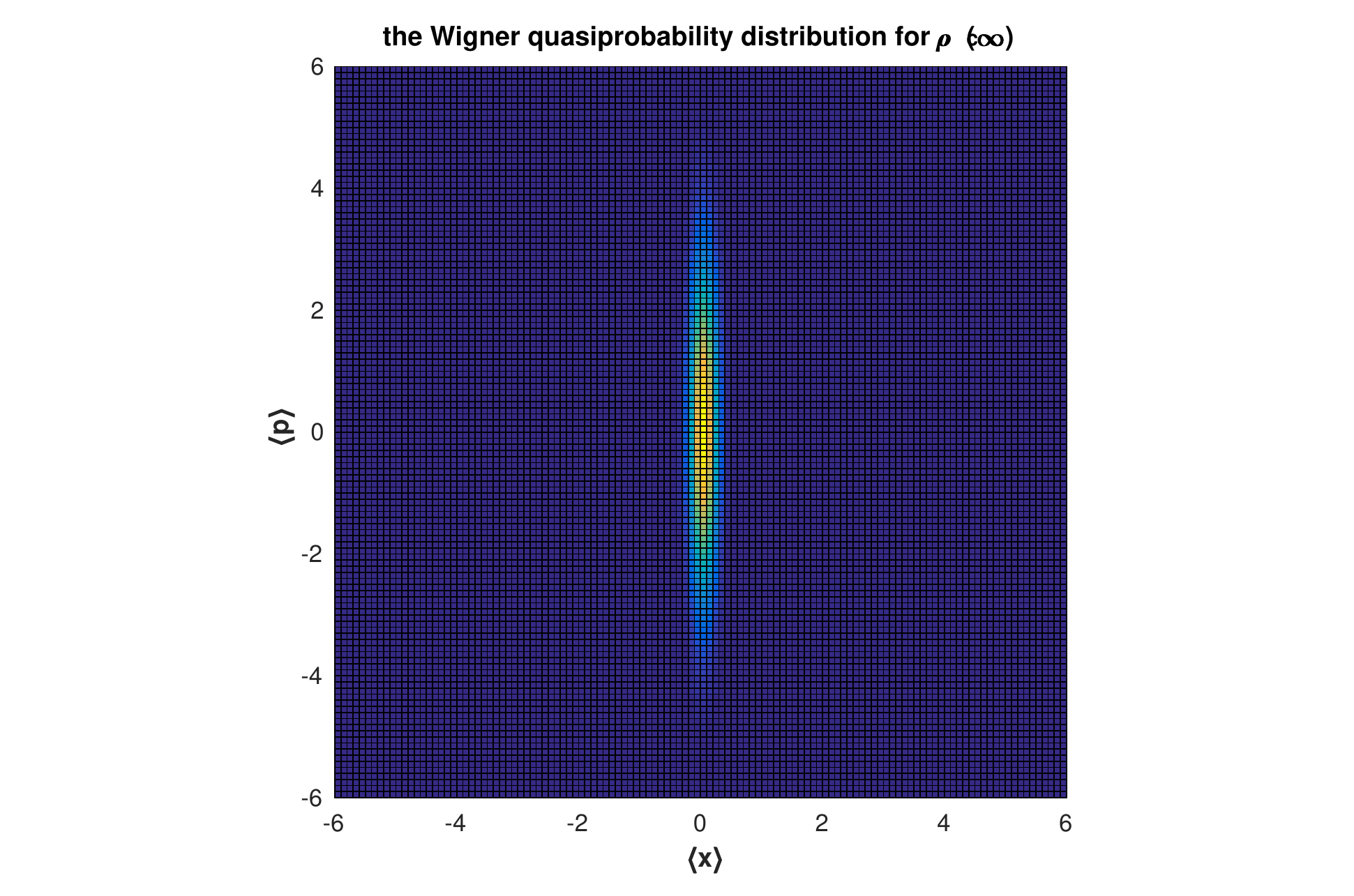}}
\put(92,177){
\tiny
\begin{tabular}{|c|}
\hline
$\q{\psi_{q^2}}$ setting \\ \hline $u=\dfrac{\pi}{4.5}\vphantom{\dfrac{A^K}{A_K}}$ \\ \hline $\mu=0$ \\ \hline $\epsilon=0$ \\ \hline $\chi=0$ \\
\hline
\end{tabular}}
\put(92,144){
\tiny
\begin{tabular}{|c|c|c|c|c|c|c|c|}
\hline 
 & $\left\langle \mathbf{X}_0 \right\rangle $ & $\left\langle \mathbf{X}_{\frac{\pi}{2}}\right\rangle $ & $\phi_r$ & $\Delta\mathbf{X}_{\frac{\phi_r}{2}}$ & $\Delta\mathbf{X}_{\frac{\phi_r+\pi}{2}}$ & $\kappa$ & $r_{\text{eff}}$(dB)\\
\hline 
\hline 
the. & 0 & 0 & $\pi$ & 1.690 & 0.148 & 0.009 & 10.6 \\  
\hline 
sim. & 0 & 0 & $\pi$ & 1.711 & 0.146 & 
& 10.7 \\
\hline 
\end{tabular}}
\put(169,156){\scriptsize Re$\alpha$}
\put(116,196){\scriptsize Im$\alpha$}


\put(122,88){\includegraphics[scale=.4, trim=45mm 8mm 10mm 8mm, clip=true]{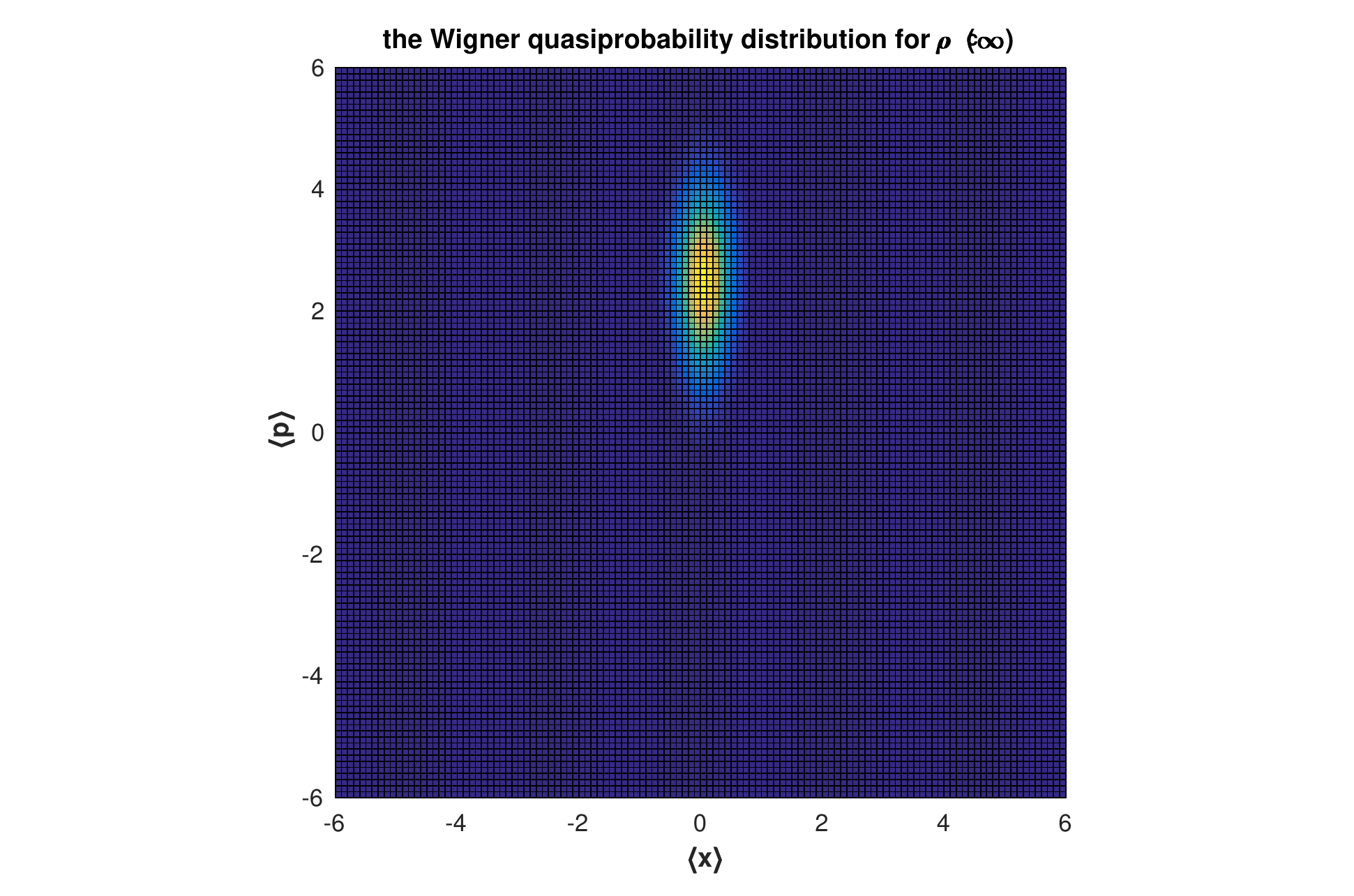}}
\put(92,110){
\tiny
\begin{tabular}{|c|}
\hline
$\q{\psi_{q^2}}$ setting \\ \hline $u=\dfrac{\pi}{6}\vphantom{\dfrac{A^K}{A_K}}$ \\ \hline $\mu=0$ \\ \hline $\epsilon=\dfrac{\pi}{30}\vphantom{\dfrac{A^K}{A_K}}$ \\ \hline $\chi=0$ \\
\hline
\end{tabular}}
\put(92,77){
\tiny
\begin{tabular}{|c|c|c|c|c|c|c|c|}
\hline 
 & $\left\langle \mathbf{X}_0 \right\rangle $ & $\left\langle \mathbf{X}_{\frac{\pi}{2}}\right\rangle $ & $\phi_r$ & $\Delta\mathbf{X}_{\frac{\phi_r}{2}}$ & $\Delta\mathbf{X}_{\frac{\phi_r+\pi}{2}}$ & $\kappa$ & $r_{\text{eff}}$(dB)\\
\hline 
\hline 
the. & 0 & 2.576 & $\pi$ & 0.966 & 0.259 & 0.025 & 5.7 \\
\hline 
sim. & 0 & 2.454 & $\pi$ & 0.921 & 0.272 & 
& 5.3 \\
\hline 
\end{tabular}}
\put(169,89){\scriptsize Re$\alpha$}
\put(116,129){\scriptsize Im$\alpha$}


\put(30,88){\includegraphics[scale=.4, trim=45mm 8mm 10mm 8mm, clip=true]{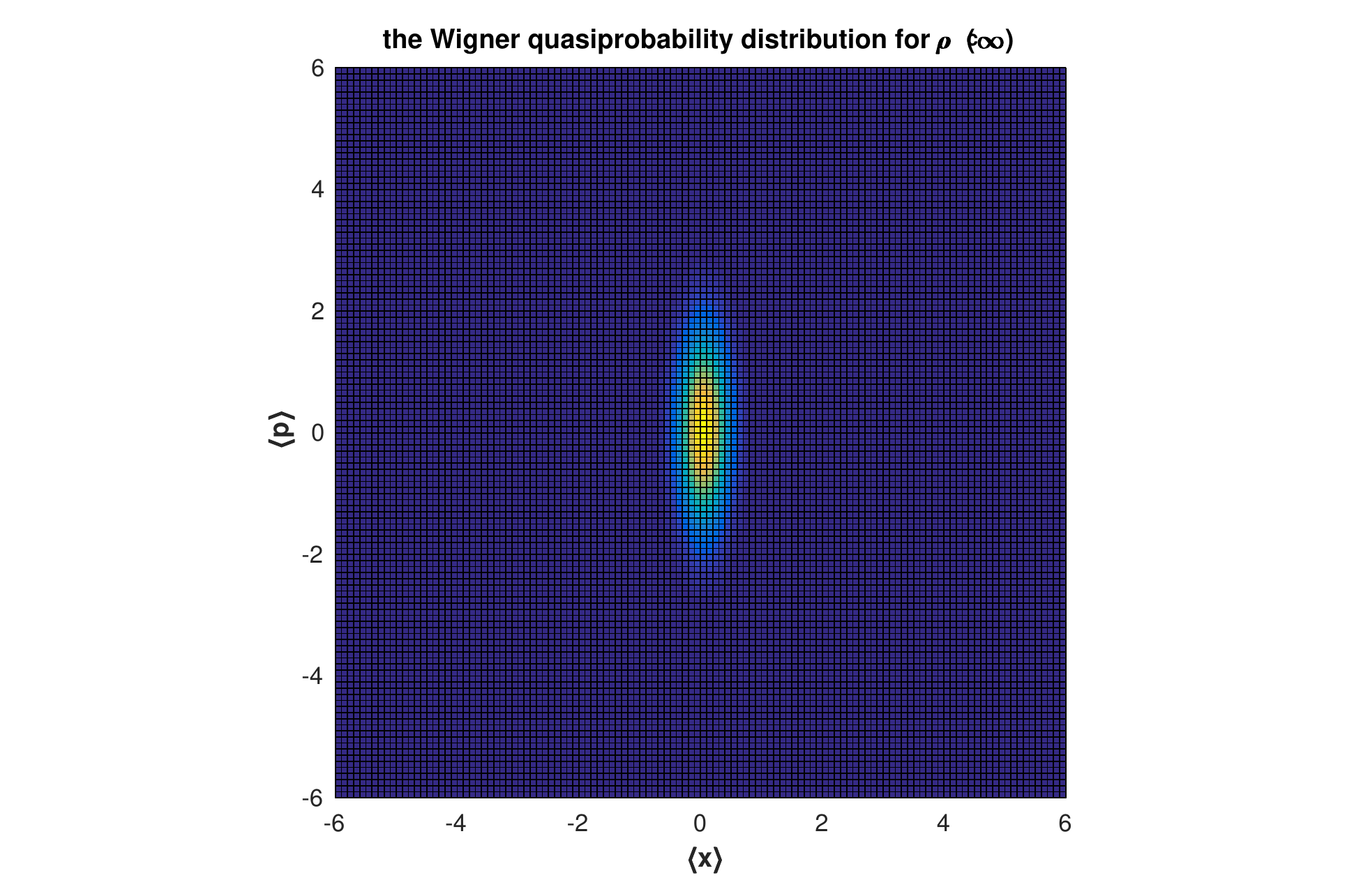}}
\put(0,110){
\tiny
\begin{tabular}{|c|}
\hline
$\q{\psi_{q^2}}$ setting \\ \hline $u=\dfrac{\pi}{6}\vphantom{\dfrac{A^K}{A_K}}$ \\ \hline $\mu=0$ \\ \hline $\epsilon=0$ \\ \hline $\chi=0$ \\
\hline
\end{tabular}}
\put(0,77){
\tiny
\begin{tabular}{|c|c|c|c|c|c|c|c|}
\hline 
 & $\left\langle \mathbf{X}_0 \right\rangle $ & $\left\langle \mathbf{X}_{\frac{\pi}{2}}\right\rangle $ & $\phi_r$ & $\Delta\mathbf{X}_{\frac{\phi_r}{2}}$ & $\Delta\mathbf{X}_{\frac{\phi_r+\pi}{2}}$ & $\kappa$ & $r_{\text{eff}}$(dB)\\
\hline 
\hline 
the. & 0 & 0 & $\pi$ & 0.966 & 0.259 & 0.025 & 5.7 \\
\hline 
sim. & 0 & 0 & $\pi$ & 0.970 & 0.258 & 
& 5.75 \\
\hline 
\end{tabular}}
\put(77,89){\scriptsize Re$\alpha$}
\put(24,129){\scriptsize Im$\alpha$}

\put(30,19){\includegraphics[scale=.4, trim=45mm 8mm 10mm 8mm, clip=true]{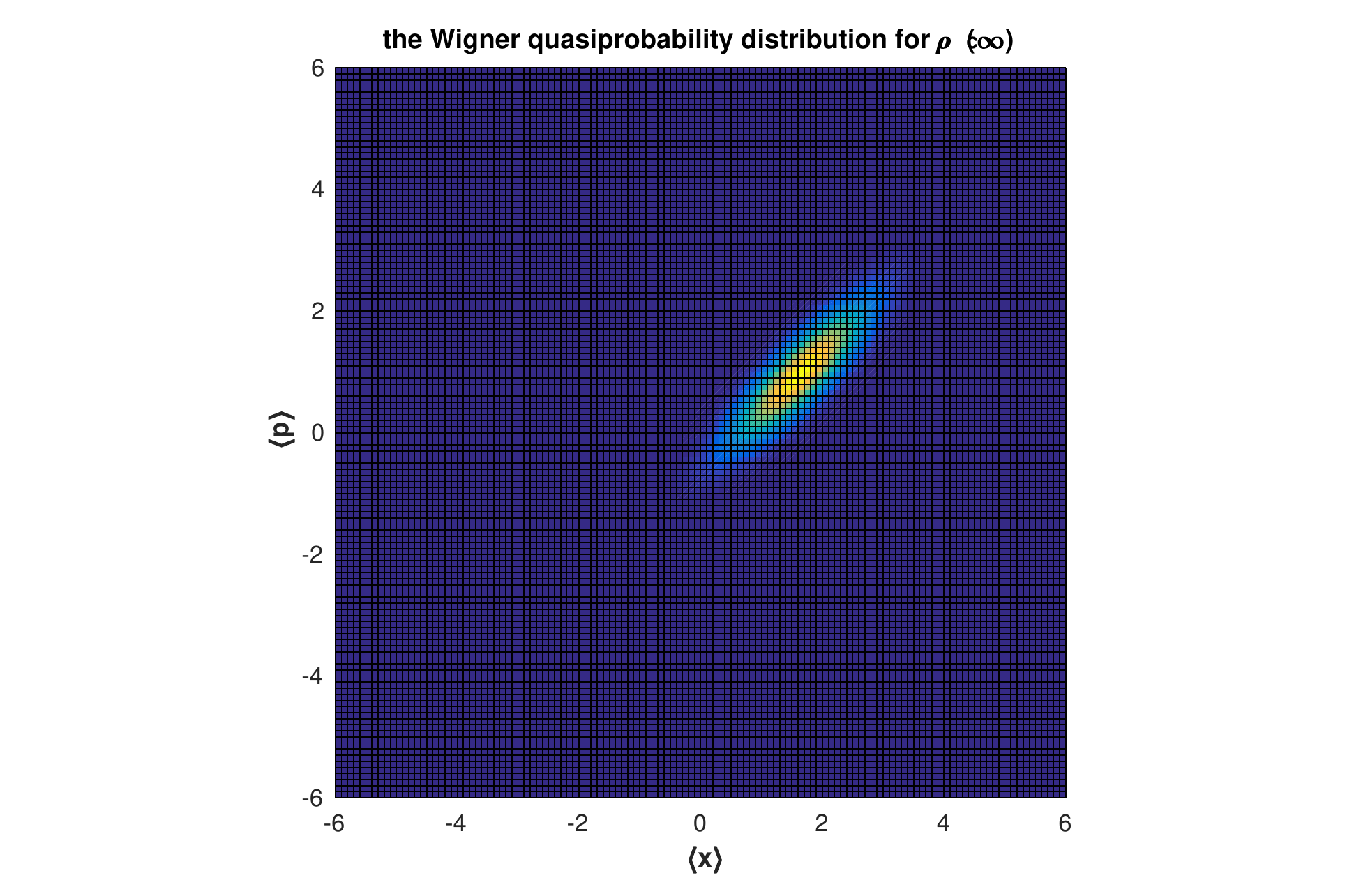}}
\put(0,43){
\tiny
\begin{tabular}{|c|}
\hline
$\q{\psi_{q^2}}$ setting \\ \hline $u=\dfrac{\pi}{6}\vphantom{\dfrac{A^K}{A_K}}$ \\ \hline $\mu=\pi/2$ \\ \hline $\epsilon=\pi/30$ \\ \hline $\chi=\pi/2$ \\
\hline
\end{tabular}}
\put(0,8){
\tiny
\begin{tabular}{|c|c|c|c|c|c|c|c|}
\hline 
 & $\left\langle \mathbf{X}_0 \right\rangle $ & $\left\langle \mathbf{X}_{\frac{\pi}{2}}\right\rangle $ & $\phi_r$ & $\Delta\mathbf{X}_{\frac{\phi_r}{2}}$ & $\Delta\mathbf{X}_{\frac{\phi_r+\pi}{2}}$ & $\kappa$ & $r_{\text{eff}}$(dB)\\
\hline 
\hline 
the. & 1.633 & 0.943 & $\frac{\pi}{2}$ & 0.966 & 0.259 & 0.025 & 5.7 \\
\hline 
sim. & 1.589 & 0.921 & $\frac{\pi}{2}$ & 0.946 & 0.265 & 
& 5.5 \\
\hline 
\end{tabular}}
\put(77,20){\scriptsize Re$\alpha$}
\put(24,60){\scriptsize Im$\alpha$}

\put(122,19){\includegraphics[scale=.4, trim=45mm 8mm 10mm 8mm, clip=true]{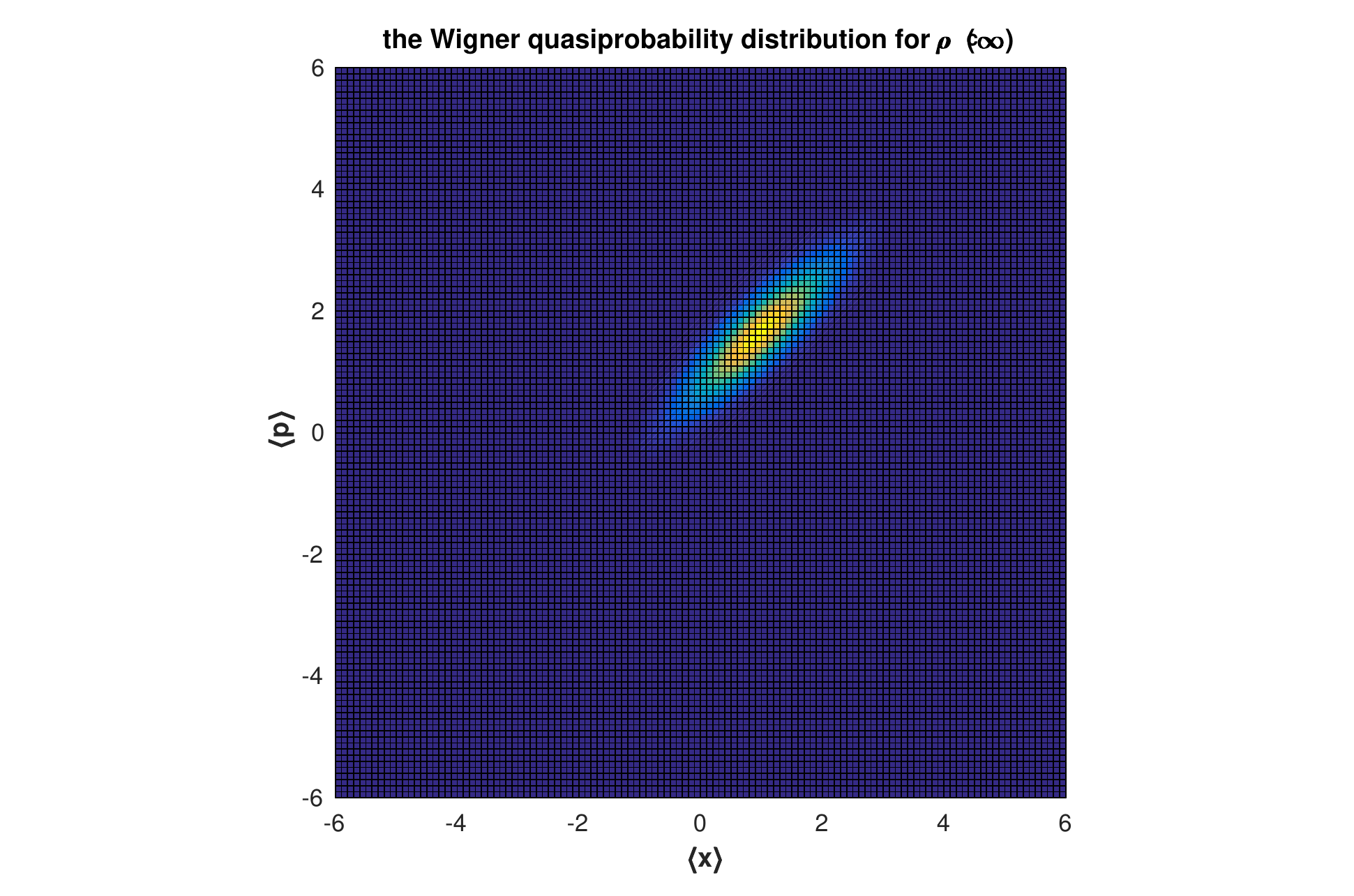}}
\put(92,43){
\tiny
\begin{tabular}{|c|}
\hline
$\q{\psi_{q^2}}$ setting \\ \hline $u=\dfrac{\pi}{6}\vphantom{\dfrac{A^K}{A_K}}$ \\ \hline $\mu=\pi/2$ \\ \hline $\epsilon=\pi/30$ \\ \hline $\chi=0$ \\
\hline
\end{tabular}}
\put(92,8){
\tiny
\begin{tabular}{|c|c|c|c|c|c|c|c|}
\hline 
 & $\left\langle \mathbf{X}_0 \right\rangle $ & $\left\langle \mathbf{X}_{\frac{\pi}{2}}\right\rangle $ & $\phi_r$ & $\Delta\mathbf{X}_{\frac{\phi_r}{2}}$ & $\Delta\mathbf{X}_{\frac{\phi_r+\pi}{2}}$ & $\kappa$ & $r_{\text{eff}}$(dB)\\
\hline 
\hline 
the. & 0.943 & 1.633 & $\frac{\pi}{2}$ & 0.966 & 0.259 & 0.025 & 5.7 \\
\hline 
sim. & 0.921 & 1.589 & $\frac{\pi}{2}$ & 0.946 & 0.265 & 
& 5.5 \\
\hline 
\end{tabular}}
\put(169,20){\scriptsize Re$\alpha$}
\put(116,60){\scriptsize Im$\alpha$}

\end{picture}
\caption{(Colour online) Wigner quasi-probability distribution for the cavity steady state $\rho_c(\infty)$, whose dynamics is governed by the Kraus map \eqref{eq:rhodyns},\eqref{eq:M2qubits}, for $\theta=\pi/20$ and various tuning of the input qubits' parameters in \eqref{eq:qubitis}. The tables provide the tuning values of $\q{\psi_{q^2}}$ as well as characteristics of $\rho_c$ according to Thm.1 and to simulations. Thanks to $\theta \ll 1$, all presented cases feature a good agreement with the theory. Note that by taking $u>0$, we have a negative value for $r$ so the minimum uncertainty is obtained on $\Delta \mathbf{X}_{\frac{\phi_r+\pi}{2}}$. In this figure, the. and sim. are short for theoretical approximations and simulations results respectively.}\label{bigfig:ideal2}
\end{figure*}
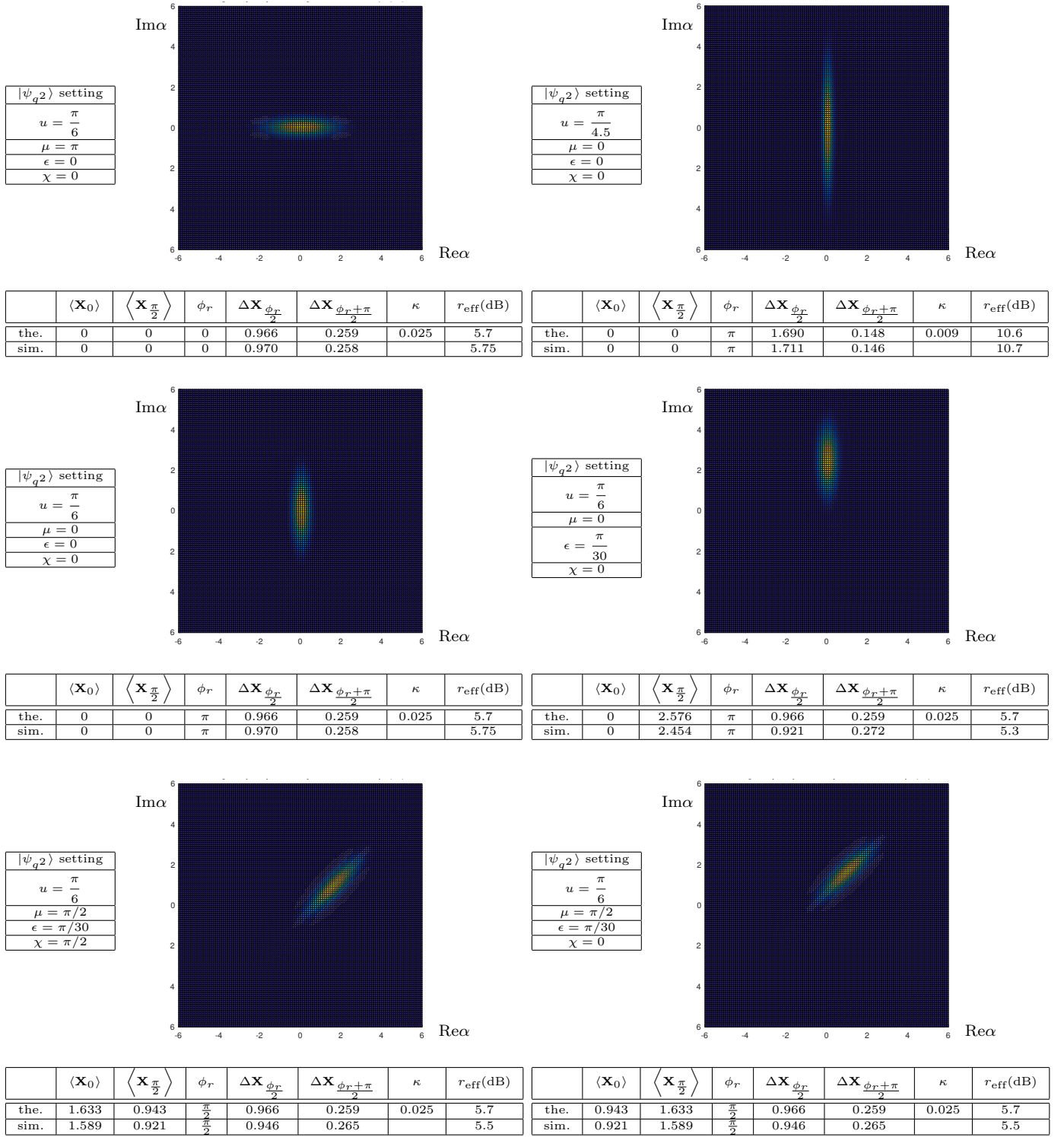

We next add two types of realistic losses to the setting.

First, we include energy decay of the harmonic oscillator in a zero temperature bath, via the annihilation operator. Over a small time $t_r$ compared to the oscillator characteristic lifetime $1/\gamma$, the corresponding effect can be modeled by the Kraus map:
$\rho_c \mapsto \mathbf{M}_0 \rho_c \mathbf{M}_0^\dagger + \mathbf{M}_1 \rho_c \mathbf{M}_1^\dagger$ with
\begin{align}
\mathbf{M}_0 &= \mathbf{I} - \frac{\gamma t_r}{2}\mathbf{a}^\dagger\mathbf{a}\, ,\nonumber\\
\mathbf{M}_1 &= \sqrt{\gamma t_r}\mathbf{a} \, .
\label{eq:onethement}
\end{align}
In the simulations, we alternate this Kraus map with the qubit interaction Kraus map \eqref{eq:rhodyns},\eqref{eq:M2qubits}.
The decay vs.~squeezing tradeoff is governed by the ratio $\Omega/\gamma$, which we take as $1000 \pi$, considered as reasonable numbers in the literature (e.g.~$\Omega=2\pi 10$ kHz and $1/\gamma = 50$ ms in a cavity QED setup \cite{SLBRR12,SRBR11}). According to Theorem 1, this tradeoff will limit the practical values of $u$, as a more squeezed steady state implies a slower squeezing rate to fight the decay. It further suggests taking larger $t_r$, since decay per interaction will increase in the same order as $t_r$, whereas convergence towards the squeezed state is in $t_r^2$ per interaction; this is only valid with the approximation of small $\theta$, but in simulations, values up to at least $\theta=\pi/4$ appear to work without trouble.

Second, we consider the possible loss of entanglement between the qubits before they interact with the oscillator. Such a loss can be due to imperfect preparation, occasional loss of a qubit before its interaction with the oscillator (see e.g.~in \cite{HR06B,DDC08,DAM93,BSM96,SDZ11}) or qubit decoherence, although we expect the latter to be quite negligible on the timescales of a single interaction. We will model this by assuming that instead of having input qubit pairs in a pure state e.g.~$\q{\psi_{q^2}} = \beta_{gg} \q{gg} + \beta_{ee} \q{ee}$ to stabilize the squeezed vacuum, we will start with a mixed state 
\begin{align*}
\rho_{q^2} = |\beta_{gg}|^2 \q{gg}\qd{gg} + |\beta_{ee}|^2 \q{ee}\qd{ee} + \eta\, (\beta_{gg}^*\beta_{ee} \q{ee}\qd{gg} +  \beta_{ee}^*\beta_{gg} \q{gg}\qd{ee})
\end{align*}
where $\eta<1$. There are various ways to decompose this mathematical expression into physically intuitive effects. One possibility is to say, with probability $1-\eta$, we just have a classical correlation, which in the Lindblad approximation (small $\theta$) corresponds to a thermal bath with $(n_{\text{thermal}})/(1+n_{\text{thermal}}) = \tan^2(u)$, i.e.~infinite temperature as $u$ approaches the optimally squeezing value $\pi/4$. Another viewpoint is that with probability $p_f = (1-\eta)/2$ we have a phase flip, in which case the entangled pair would stabilize a squeezed state turned by $\pi/2$, i.e.~in the wrong direction. The corresponding Lindblad approximation would thus just be
\begin{align*}
\tfrac{d}{d\tau}\rho_c = \kappa\; \left((1\text{-}p_f)\mathcal{D}(\mathbf{S}(r)\mathbf{a}\mathbf{S}^\dagger(r)) + p_f \mathcal{D}(\mathbf{S}(\text{-}r)\mathbf{a}\mathbf{S}^\dagger(\text{-}r))\right) \, \rho_c.
\end{align*}
This illustrates the fact that, in the absence of oscillator decay, $\theta$ has no effect on the associated steady state, which can be computed by writing the dynamics of $\langle \mathbf{X}_0^2 \rangle$ and $\langle \mathbf{X}_{\frac{\pi}{2}}^2 \rangle$ in the Heisenberg picture associated to this Lindblad equation. This takes the closed form:
\begin{eqnarray}\label{eq:imp}
\tfrac{d}{dt} \langle \mathbf{X}_0^2 \rangle &=& 2 \theta^2 \Big( -\cos2u \langle \mathbf{X}_0^2 \rangle + \tfrac{1}{4}(1+\eta\sin 2u) \Big), \\
\nonumber \tfrac{d}{dt} \langle \mathbf{X}_{\frac{\pi}{2}}^2 \rangle &=& 2 \theta^2 \Big( -\cos2u \langle \mathbf{X}_{\frac{\pi}{2}}^2 \rangle + \tfrac{1}{4}(1-\eta\sin 2u) \Big) .
\end{eqnarray}
The optimal tuning according to this equation is $\sin 2u= \eta$, which yields an optimal
$$\Delta \mathbf{X}_{\frac{\pi}{2}} = \sqrt{\langle \mathbf{X}_{\frac{\pi}{2}}^2\rangle} = \, \frac{1}{2}\, (1-\eta^2)^{1/4} \; . $$
Note that \eqref{eq:imp} also expresses the absence of benefit with just classical qubit correlations, i.e.~with $\rho_{q^2} = |\beta_{gg}|^2 \q{gg}\qd{gg} + |\beta_{ee}|^2 \q{ee}\qd{ee}$ and thus $\eta=1-2p_f=0$: the uncertainties are $1/(2\cos2u)$ and thus take their minimum with ground qubits and vacuum oscillator. (We do not add the oscillator decay to this Lindblad study, even though it would be easy to do so, because it would not really help us find the best squeezing: the main limitation on cranking up $\theta$, which helps counter oscillator decay, is to invalidate the Lindblad approximation.)

With this setting, qubit decoherence prior to interaction appears to be the most constraining imperfection. For instance, according to the formula, a value of $\eta=99.5\%$ would be necessary in order to obtain the $10$dB benchmark of squeezing \cite{VMC08} on the vacuum state. Simulations, illustrated in Figure \ref{fig:optim}, confirm that we then reach, at best, $9.4$dB, with $u=\pi/4.3$ and $\theta = \pi/28$. The oscillator decay has a negligible effect and the system is not very sensitive to $\theta$, except for the most extreme values of $u$. This suggests that implementing entanglement purification strategies like progressive stabilization or heralding \cite{MBD16}, even at the cost of reducing the number of qubit pairs per second, would be beneficial with this type of setting.


\begin{figure}[!htp]
\centering
\includegraphics[width=95mm, trim=10mm 5mm 0mm 5mm, clip=true]{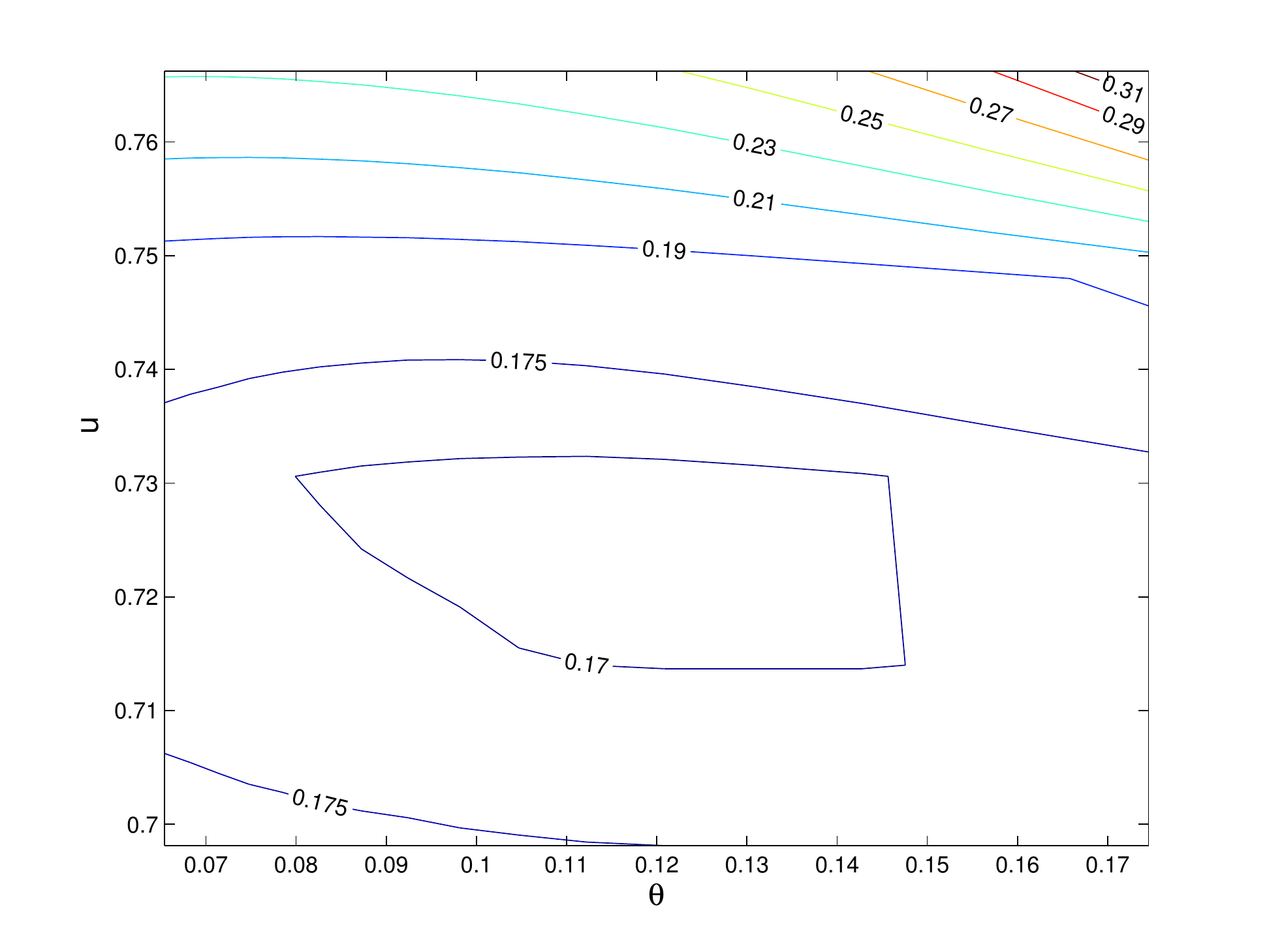}
\caption{(color online) Squeezing of the steady state, described by $\Delta \mathbf{X}_{\frac{\pi}{2}}$, for various tunings of our entangled-qubit-pair reservoir in presence of imperfections (see main text). We have focused on the region of $\theta$ small and $u$ close to $\pi/4$, where squeezing is maximal. Any value below $0.5$ implies squeezing, and the optimum is found to be $\Delta \mathbf{X}_{\frac{\pi}{2}} \simeq 0.169$ corresponding to $9.4$dB, with $u=\pi/4.3$ and $\theta = \pi/28$.}\label{fig:optim}
\end{figure}


\section{Resonant interaction with many entangled qubits}\label{sec:interp}

We now proceed with our more theoretically motivated analysis of the effect of time-entanglement in a stream of reservoir qubits. The observations from the previous section raise two questions.
\newline (i) Can we ``win'' with tailored inputs involving some optimally entangled set of $p>2$ qubits, e.g.~stabilize a more squeezed steady state than with $p=2$?
\newline (ii) Can we consider a physical setup where entanglement is continuously created by the stream of qubits passing through a time-invariant device, before they interact with the oscillator? 

From the simulations with imperfections described in the previous section, the first question suggests that the corresponding experiments would be difficult to carry out in practice. Besides the necessity to generalize a CNOT entangling device towards some more complicated circuit, the apparent danger of destroying all squeezing benefits when a little coherence is lost does not bode well for such tailored joint states of the qubits. We therefore leave the first question for future work and concentrate on the second one, which appears experimentally clearer and naturally involves only local entanglement.


\subsection{A continuously entangled stream}

This more operational viewpoint on time-entangled reservoirs, keeping a continuous-time limit in mind, would work as depicted in Figure \ref{fig:entangler}. We will assume throughout this section that the qubits enter the setup in their ground state $\q{g}$. We place upstream of the oscillator an ``entangler'' device/gate that enacts a joint unitary $\mathbf{U}_E$ onto a pair of qubits, e.g.~
\begin{equation}\label{eq:entangler}
\mathbf{U}_E = \text{exp}\left(-i \phi \sigma_x \otimes \sigma_y \right) = \cos\phi \mathbf{I} - i \sin\phi \sigma_x \otimes \sigma_y \, .
\end{equation}
By applying $\mathbf{U}_E$ once every \emph{two time steps} while the qubits pass sequentially through this device, we would obtain a stream of pairwise entangled qubits, exactly as used in Section \ref{sec:pairs}; indeed, $\mathbf{U}_E\,\q{gg} = \beta_{gg} \q{gg} + \beta_{ee} \q{ee}$ with $\beta_{gg},\beta_{ee}$ real. Here instead, we will operate the device once at \emph{every discrete time step}, such that it acts twice on each qubit $t$: once to entangle it with qubit $t-1$ and once to entangle it with $t+1$. The resulting continuously entangled stream of qubits is then sent to the oscillator, to undergo as previously e.g.~the resonant interaction \eqref{eq:ResInt} and then be dropped.

This operational viewpoint allows for a scalable analysis of the reservoir setup. Since operations on disjoint Hilbert spaces commute, the precise physical position of the entangler (and qubit dropping) --- just before the oscillator, or with several travelling qubits between entangler and oscillator --- has no effect on the resulting reservoir.

\begin{figure}[!htp]
\centering
\includegraphics[scale=.6]{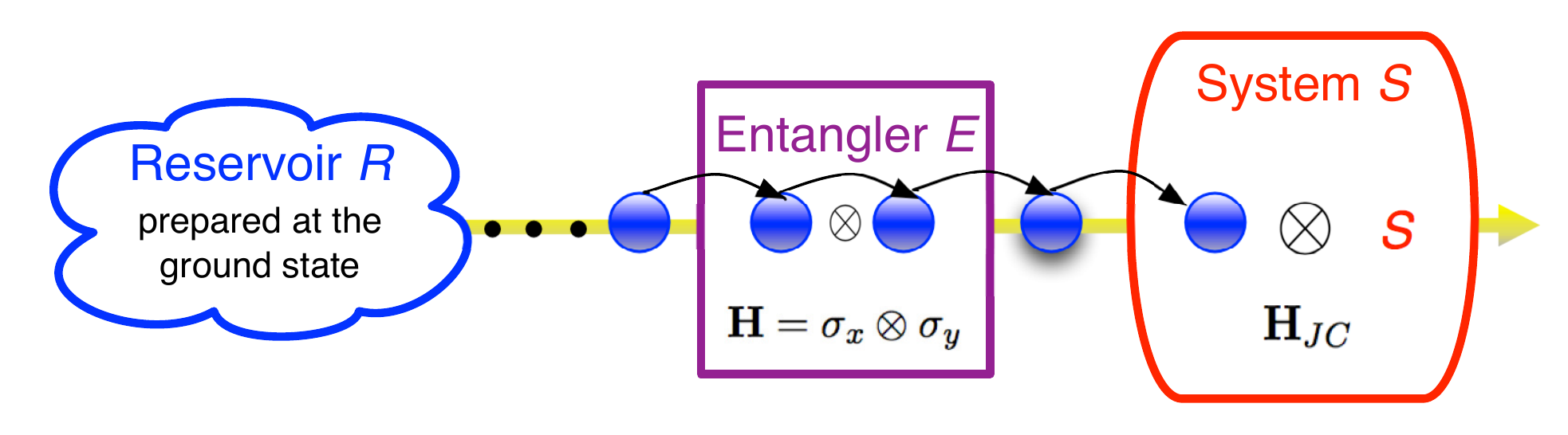}
\caption{(Colour online) Scheme of the reservoir with recursively entangled qubits. The stream of reservoir qubits, initially in the ground state, is sent through an ``entangler'' where each qubit is successively entangled with the qubit directly before and directly after it in time. The operation of the entangler is characterized by the propagator $\mathbf{U}_{E}$ as per \eqref{eq:entangler}, while resonant interaction between the oscillator and each qubit is given by the propagator $\mathbf{U}_{r}$ as per \eqref{eq:ResInt}.}
\label{fig:entangler}
\end{figure}

When placing the entangler far upstream, we can examine the resulting joint qubits state, which is in a Matrix Product State form \cite{RO14}. This form is in some sense ``locally'' entangled in time, as the order of the qubits specifically matters (except for trivial values of $\phi$). For instance, the state of five qubits initialized in $\q{ggggg}$ after passing through this device, becomes
\begin{eqnarray*}
\beta_1 \q{ggggg} &+& \beta_2 (\q{gggee}+\q{ggeeg}+\q{geegg}+\q{eeggg}) \\
&+& \beta_3 (\q{egegg} + \q{gegeg} + \q{ggege})\\
&& + \beta_3 (\q{eeeeg}+\q{geeee}) + ...  
\end{eqnarray*}
with $\beta_1 > \beta_2 > \beta_3 > ...$ for small $\phi$. We here see a stronger weight on states with two consecutive excitations than with two non-consecutive ones, and so on.

This viewpoint may suggest a hard non-Markovian analysis for the oscillator state's evolution. Yet in fact, keeping a memory of one qubit in joint state with the oscillator, is sufficient in order to recover a Markovian process. Indeed, to analyze the reservoir effect, both analytically and in simulations, it is more efficient to place the entangler (and qubit dropping) just next to the oscillator. In fact, qubit $t+1$ does not have to be modeled before its entanglement with qubit $t$, which itself can happen just before $t$ interacts with the oscillator; at that stage, qubit $t-1$ can already have interacted and been dropped. Note that the oscillator can also be entangled with qubit $t$ before interacting with it, because it has interacted with qubit $t-1$ that was entangled with $t$. Therefore, denoting by $\rho_J(t)$ the joint state of qubit $t$ and the oscillator just \emph{before} their interaction, one time step can be described by:
\begin{eqnarray}\label{eq:rhoJdyns}
\rho_J(t) & \mapsto & \q{g}\qd{g} \otimes \rho_J(t) =: \rho'_A(t) \\ \nonumber
& \mapsto & (\mathbf{U}_E \otimes \mathbf{I}_c)\, \rho'_A(t)\, (\mathbf{U}_E \otimes \mathbf{I}_c)^\dagger =: \rho'_B(t)\\ \nonumber
& \mapsto & (\mathbf{I}_q \otimes \mathbf{U}_r)\, \rho'_B(t)\, (\mathbf{I}_q \otimes \mathbf{U}_r)^\dagger =: \rho'_C(t)\\ \nonumber
& \mapsto & \text{Trace}_{q(t)}(\rho'_C(t)) = \rho_J(t+1) \; .
\end{eqnarray}
Here we have adopted the order convention (qubit $t+1$) $\otimes$ (qubit $t$) $\otimes$ (oscillator); $\mathbf{I}_c$ and $\mathbf{I}_q$ denote identity operators on oscillator and qubit states, respectively; and $\text{Trace}_{q(t)}(\cdot)$ is the partial trace with respect to qubit $t$. Replacing $\mathbf{U}_E,\; \mathbf{U}_R$ by their expressions and factoring out $\q{g}_t$ and $\q{e}_t$ to take the partial trace, as in \eqref{eq:factorit2}, we obtain the Kraus map:
\begin{eqnarray*}
\rho_J(t+1) & = & \mathbf{M}_g \rho_J(t) \mathbf{M}_g^\dagger + \mathbf{M}_e \rho_J(t) \mathbf{M}_e^\dagger \; , \;\; \text{with}\\[2mm]
\mathbf{M}_g & = & (\cos\phi \q{g}_{t+1} \qd{g}_t + \sin\phi \q{e}_{t+1} \qd{e}_t) \cos\theta_{\mathbf{N}} \\
& & +(\cos\phi \q{g}_{t+1} \qd{e}_t + \sin\phi \q{e}_{t+1} \qd{g}_t) \tfrac{\sin\theta_{\mathbf{N}}}{\sqrt{\mathbf{N}}} \mathbf{a}^\dagger,  \\[2mm]
\mathbf{M}_e & = & (\cos\phi \q{g}_{t+1} \qd{e}_t + \sin\phi \q{e}_{t+1} \qd{g}_t) \cos\theta_{\mathbf{N+I}} \\
& & -(\cos\phi \q{g}_{t+1} \qd{g}_t + \sin\phi \q{e}_{t+1} \qd{e}_t) \, \mathbf{a} \, \tfrac{\sin\theta_{\mathbf{N}}}{\sqrt{\mathbf{N}}} \, .
\end{eqnarray*}
This dynamics expresses the transition from a joint state of oscillator and qubit $t$, to the joint state of oscillator and qubit $t+1$. It is however perfectly licit to consider the map on the joint state of ``oscillator and active qubit'', i.e.~identify the Hilbert spaces of qubit $t$ before and $t+1$ after this time step, and just drop the indices on the qubits.


\subsection{Analysis of cavity behavior}

A first intuition about the system behavior can be obtained by regrouping terms:
\begin{eqnarray}\label{eq:ghmmodel}
\mathbf{M}_g & = & (\cos\phi \q{g} \text{+} \sin\phi \q{e}) \qd{+} \, \frac{1}{\sqrt{2}} \left(\cos\theta_{\mathbf{N}} \text{+} 
\tfrac{\sin\theta_{\mathbf{N}}}{\sqrt{\mathbf{N}}} \mathbf{a}^\dagger \right) \\
\nonumber && + (\cos\phi \q{g} \text{-} \sin\phi \q{e}) \qd{-} \, \frac{1}{\sqrt{2}} \left(\cos\theta_{\mathbf{N}} \text{-} \tfrac{\sin\theta_{\mathbf{N}}}{\sqrt{\mathbf{N}}} \mathbf{a}^\dagger \right),\\[2mm]
\nonumber \mathbf{M}_e & = & (\cos\phi \q{g} \text{+} \sin\phi \q{e}) \qd{+} \, \frac{1}{\sqrt{2}}\left(\cos\theta_{\mathbf{N+I}} \text{-} \mathbf{a} \, \tfrac{\sin\theta_{\mathbf{N}}}{\sqrt{\mathbf{N}}} \right) \\
\nonumber && + (\cos\phi \q{g} \text{-} \sin\phi \q{e}) \qd{+} \, \frac{1}{\sqrt{2}}\left(\text{-}\cos\theta_{\mathbf{N+I}} \text{-} \mathbf{a} \, \tfrac{\sin\theta_{\mathbf{N}}}{\sqrt{\mathbf{N}}} \right) \, .
\end{eqnarray}
Here we have introduced the orthonormal basis $\q{\pm} = (\q{g}\pm \q{e})/\sqrt{2}$ for the qubit. This highlights that the operator acting on the part of the state conditioned on a $\q{+}$ qubit state (resp.~a $\q{-}$ qubit state) corresponds to the result of a resonant interaction with a qubit initially in $\q{+}$ (resp.~$\q{-}$), that is roughly a displacement along the positive direction of $\mathbf{X}$ (resp.~$\mathbf{-X}$). However, for $\phi \neq \pi/4$, the resulting cavity states are not kept independent but interfere, as the qubit goes from $\q{+}$ and $\q{-}$ respectively to $(\cos\phi \q{g} \text{+} \sin\phi \q{e})$ and $(\cos\phi \q{g} \text{-} \sin\phi \q{e})$ which can have a significant overlap. Part of the oscillator state that has moved in the $\mathbf{-X}$ direction would thus interfer with the $\mathbf{+X}$-moved state, and moreover in the next iteration, part of it moves further in the positive $\mathbf{X}$ direction while part goes back along the $\mathbf{-X}$ direction. This suggests a state that is elongated along the $\mathbf{X}$ axis, with the interference possibly leading to a squeezing effect. 

A quantitative analysis of the squeezing effect is easier by writing the model fully in the $\q{+},\q{-}$ basis for the qubit: 
\begin{eqnarray*}
\rho_J  =  \rho_{+,+}  \q{+}\q{+} + \rho_{-,-}  \q{-}\q{-}+ \rho_{+,-}  \q{+}\q{-} + \rho_{+,-}^\dagger  \q{-}\q{+} \; .
\end{eqnarray*}
Indeed, in \eqref{eq:ghmmodel} we observe that swapping $\q{+}$ for $\q{-}$ and at the same time turning the cavity state by $\pi$ in its phase space (in other words applying operator $e^{i\pi \mathbf{N}}$), yields an invariance of the dynamics. From there we get that
\begin{eqnarray}\label{eq:sysms}
\rho_{+,+} = e^{i\pi \mathbf{N}}\, \rho_{-,-}\, e^{i\pi \mathbf{N}} &=:& \rho_D/2\;, \\ \nonumber
\rho_{+,-}\, e^{i\pi \mathbf{N}} = e^{i\pi \mathbf{N}} \, \rho_{+,-}^\dagger &=:& \rho_O/2 \; ,
\end{eqnarray}
with both $\rho_D$ and $\rho_O$ Hermitian, and $\rho_D$ positive semi-definite of trace $1$. From this we can reduce the dynamics to the evolution of $\rho_{D}$ and $\rho_O$:
\begin{eqnarray}
\label{eq:1rhoD}
\rho_D(t+1) & = & \tfrac{1\text{+}\sin2\phi}{2} \Phi(\rho_D(t)) \\
\nonumber && + \tfrac{1\text{-}\sin2\phi}{2} e^{i\pi \mathbf{N}}\Phi(\rho_D(t))e^{i\pi \mathbf{N}}\\
\nonumber && + \tfrac{\cos2\phi}{2} \left(e^{i\pi \mathbf{N}} \, \Upsilon(\rho_O(t)) +  \Upsilon(\rho_O(t)) \,  e^{i\pi \mathbf{N}} \right),\\
\label{eq:2rhoO}
\rho_O(t+1) & = & \tfrac{1\text{+}\sin2\phi}{2} \Upsilon(\rho_O(t)) \\
\nonumber && + \tfrac{1\text{-}\sin2\phi}{2} e^{i\pi \mathbf{N}}\Upsilon(\rho_O(t))e^{i\pi \mathbf{N}}\\
\nonumber && + \tfrac{\cos2\phi}{2} \left(e^{i\pi \mathbf{N}} \, \Phi(\rho_D(t)) +  \Phi(\rho_D(t)) \,  e^{i\pi \mathbf{N}} \right)
\end{eqnarray}
where we have defined two superoperators. The first one, 
\begin{eqnarray*}
\Phi(\rho) & = & \frac{1}{2}(\cos\theta_{\mathbf{N}} + \tfrac{\sin\theta_{\mathbf{N}}}{\sqrt{\mathbf{N}}} \mathbf{a}^\dagger) \rho (\cos\theta_{\mathbf{N}} + \tfrac{\sin\theta_{\mathbf{N}}}{\sqrt{\mathbf{N}}} \mathbf{a}^\dagger)^\dagger \\
&& + \frac{1}{2}(\cos\theta_{\mathbf{N+1}} - \mathbf{a} \tfrac{\sin\theta_{\mathbf{N}}}{\sqrt{\mathbf{N}}}) \rho (\cos\theta_{\mathbf{N+1}} - \mathbf{a} \tfrac{\sin\theta_{\mathbf{N}}}{\sqrt{\mathbf{N}}})^\dagger \; ,
\end{eqnarray*}
corresponds to a trace-preserving completely positive map, equivalent to the effect of resonant interaction with a single qubit initialized in the $\q{+}$ state. The second one,
\begin{eqnarray*}
\Upsilon(\rho) & = & \frac{1}{2}(\cos\theta_{\mathbf{N}} + \tfrac{\sin\theta_{\mathbf{N}}}{\sqrt{\mathbf{N}}} \mathbf{a}^\dagger) \rho (\cos\theta_{\mathbf{N}} + \tfrac{\sin\theta_{\mathbf{N}}}{\sqrt{\mathbf{N}}} \mathbf{a}^\dagger)^\dagger \\
&& - \frac{1}{2}(\cos\theta_{\mathbf{N+1}} - \mathbf{a} \tfrac{\sin\theta_{\mathbf{N}}}{\sqrt{\mathbf{N}}}) \rho (\cos\theta_{\mathbf{N+1}} - \mathbf{a} \tfrac{\sin\theta_{\mathbf{N}}}{\sqrt{\mathbf{N}}})^\dagger \; ,
\end{eqnarray*}
is not trace-preserving and basically expresses the interference discussed around equation \eqref{eq:ghmmodel}.

Our aim then is to compute the steady-state values of the cavity state for small values of $\theta$, i.e.~very short interaction time between the cavity and each qubit in the stream. In appendix, we explain how this equation can be analyzed to obtain the following results at steady state for $\rho_D$:
\begin{itemize}
\item $\langle \mathbf{X}_{\pi/2} \rangle = 0 \; ,$   
\item $\langle \mathbf{X}_{0} \rangle = \dfrac{\theta \sin 2\phi}{2\,(1-\sin 2\phi)} + O(\theta^3) \; ,$
\item $\Delta \mathbf{X}_{\pi/2} = \sqrt{\dfrac{1-2\sin2\phi\cos^2 2\phi}{4 \cos^2 2\phi}} + O(\theta)\; .$
\end{itemize}
Using \eqref{eq:sysms}, this implies the same steady state values of $\langle \mathbf{X}_{\pi/2} \rangle$ and $\Delta \mathbf{X}_{\pi/2}$ for the oscillator mode steady state, resulting from tracing out the qubit in $\rho_J$. Since $\rho_{+,+}$ and $\rho_{-,-}$ will have opposite values of $\langle \mathbf{X}_{0} \rangle$, this indicates some mixing due to the entanglement with the qubit, which, however, vanishes as $\theta$ becomes small. The first-order squeezing expressed in $\Delta \mathbf{X}_{\pi/2}$ does not depend on $\theta$, but it depends heavily on $\phi$, as depicted in Figure \ref{fig:DelpiFromPhi}. For $\phi=0$ we get zero squeezing; for $\pi/4$ the uncertainty becomes infinite as basically the state drifts off and our analysis is invalidated; the situation is not symmetric with the sign of $\phi$; and a minimum uncertainty is obtained at $\phi \simeq 0.2763$, yielding $\Delta \mathbf{X}_{\pi/2} \simeq 0.2874$ which corresponds to about 5dB of squeezing.

This indicates that entanglement in a stream of qubits does stabilize a related squeezed state, although bounded by some constant value. Furthermore, in the extreme case where both interaction times ($\theta$ and $\phi$) go to zero, the entanglement has no effect (at least on the observables mentioned above). We have run some simulations with the full exact model to indeed confirm these conclusions. The Wigner function of the oscillator state takes the shape of a squeezed state, not saturating the Heisenberg uncertainty at finite $\theta$, with dissymetry in $\phi$, and with the values of $\langle \mathbf{X}_{\pi/2} \rangle$, $\langle \mathbf{X}_{0} \rangle$, $\Delta \mathbf{X}_{\pi/2}$ closely matching the ones computed above.

\begin{figure}
\centering
\includegraphics[width=0.9\columnwidth, trim=15mm 0mm 15mm 5mm, clip=true]{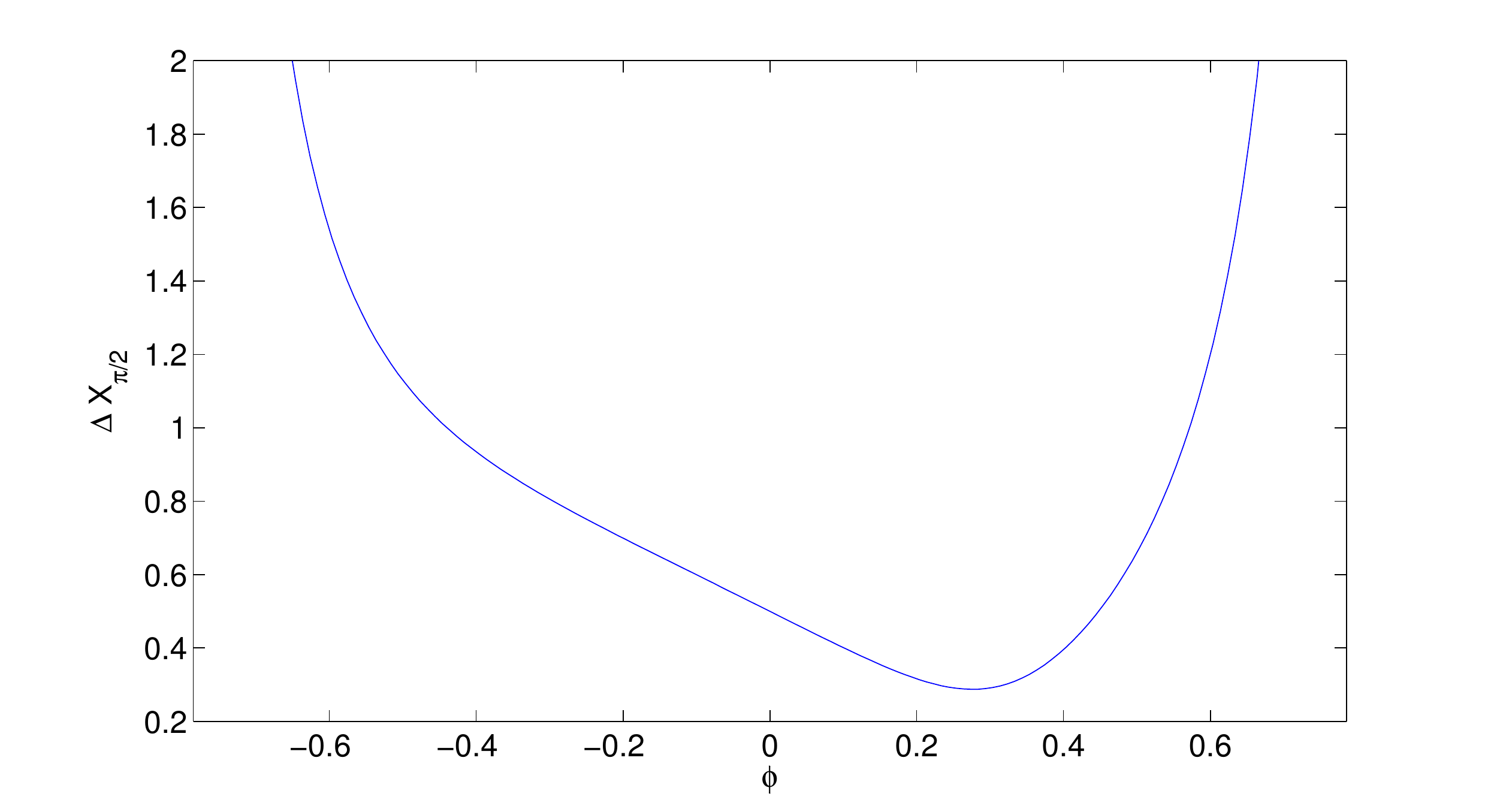}
\caption{Steady-state value of $\Delta \mathbf{X}_{\pi/2}$ for the oscillator state (weakly entangled qubit traced out) as a function of entangler parameter $\phi$, for small values of qubit-oscillator interaction $\theta$.}\label{fig:DelpiFromPhi}
\end{figure}


\section{Conclusion}

The main motivation for this paper is a theoretical analysis of some possible consequences of time-entanglement in reservoir inputs. We have shown that a quantum effect appears, by proving how the entanglement among a stream of qubits, consecutively interacting with a resonant harmonic oscillator, stabilizes a strongly squeezed state whose squeezed quadrature orientation depends on the entanglement phase. By taking a very weakly interacting reservoir ($\theta$ going to zero), the stream of entangled qubits can be pushed to the continuous-time limit. 

In terms of experimental realizations, one can of course imagine various variants to the cavityQED setup of \cite{DDC08} considered here for pedagogical reasons. For instance, a straightforward circuitQED architecture would consist of a high-Q cavity coupled to a pair of transmon qubits 1 and 2. After entangling those qubits and switching on consecutively their respective interactions with the cavity, they can be fast reset e.g.~with an auxiliary low-Q cavity as described in \cite{GLP13}. The same physical qubits can then act as qubits 3 and 4 of the scheme described in Section III. This is also sufficient to implement the proposal of Section IV, which likewise has only 2 ``active'' qubits at any time. The first qubit can be reset directly after its interaction with the high-Q cavity as it has no further use, and is now playing the role of effective qubit 3. Then it undergoes the entanglement operation with the ``still active'' qubit 2 before qubit 2 interacts with the high-Q cavity, and so on. We can also mention that, although we know of no such existing experimental setup, in principle nothing requires the field to be trapped and the qubits moving: one could as well imagine a field travelling through a lattice of qubits. The latter would then have to be set up in a state that stabilizes the desired state for the propagating field, e.g.~entangled lattice in order to generate squeezing. In a converse way, from our observations, one could view the amount of squeezing in such setup as a witness of entanglement in a stream or lattice of qubits.

The observations in this paper motivate further work on the effects of entanglement in the inputs to quantum systems. From a principles viewpoint, one may want to prove a ``quantum-improved'' scaling, e.g.~that a particular state of $m$ entangled qubits leads to better performance than $m/2$ pairwise entangled ones. In preliminary calculations, we have confirmed how the obvious GHZ state is not a good candidate and rather decreases the steady-state squeezing, since before the last qubit has entered, the setup has essentially been a thermal bath (undefined entanglement phase). Some optimization procedure might however still lead to an improved situation w.r.t.~pairwise entanglement. From a more practical viewpoint, this observation calls for a more systematic and careful investigation of the possible role played by time entanglement in (typically continuous-time) inputs to quantum systems. The discrete-time method developed in the present paper might suggest a tractable approach in this sense. In particular, we see that the level of entanglement ($\phi$) among consecutive qubits has to be maintained bounded away from zero to keep a squeezing effect in the small-$\theta$ limit, possibly indicating more general conclusions about a meaningful time quantization of such entangled streams.

\section*{Acknowledgments}
The authors would like to give special thanks to P. Rouchon for drawing our attention to this open question. The authors would also like to thank M. Mirrahimi, Z. Leghtas and B. Huard for stimulating discussions.


\section*{Appendix}

The analysis of \eqref{eq:1rhoD},\eqref{eq:2rhoO} for small $\theta$ exploits an iterative procedure. The key point is to note that $\Phi(\cdot)$ gives a result of order $1$, while $\Upsilon(\cdot)$ gives a result of order $\theta$. We further use a Taylor series expansion for $\cos\theta_{\mathbf{N}}$ and $\sin\theta_{\mathbf{N}}$, and obvious relations about the commutation of operators with $e^{i\pi{\mathbf{N}}}$. We then just write the steady-state conditions and obtain information from them.

$\bullet$ We start with \eqref{eq:1rhoD} at 0 order in $\theta$, where $\Phi$ is almost identity. The trick is that $\rho_O$, even if unknown, anyway appears only at $O(\theta)$, so we obtain:
$$e^{i\pi{\mathbf{N}}}\, \rho_D\, e^{i\pi{\mathbf{N}}} = \rho_D + O(\theta)\, .$$
Now plugging this into \eqref{eq:2rhoO}, we get
$$\rho_O = \cos2\phi\, \rho_D + O(\theta)\, .$$

$\bullet$ Now that we know $\rho_O$ up to order $\theta$, this would yield only an $O(\theta^2)$ uncertainty in \eqref{eq:1rhoD} and we can refine to obtain:
$$e^{i\pi{\mathbf{N}}}\, \rho_D\, e^{i\pi{\mathbf{N}}} = \rho_D + \frac{\theta \sin2\phi}{1-\sin2\phi}\, [(\mathbf{a}-\mathbf{a}^\dagger),\, \rho_D]\; +  O(\theta^2)\, .$$
The right-hand side corresponds to a first-order Taylor expansion of the displacement operator $\mathbf{D}\left(-\alpha\right)$ acting on $\rho_D$, with $\alpha = \frac{-\theta \sin2\phi}{1-\sin2\phi} \, .$ With this we can already compute $\langle \mathbf{X}_{\pi/2} \rangle$ and $\langle \mathbf{X}_{0} \rangle$ to first order. This confirms the $0$ on $\langle \mathbf{X}_{\pi/2} \rangle$, which should be exact stemming from symmetry arguments, and gives the value provided for $\langle \mathbf{X}_{0} \rangle$.
\newline Plugging this into \eqref{eq:2rhoO} gives an expression for $\rho_O$ up to order $O(\theta^2)$ terms; we don't provide it here as $\rho_O$ is not our main object of interest.

$\bullet$ Iterating once more, we plug the $\rho_O$ with $O(\theta^2)$ uncertainty into \eqref{eq:1rhoD} in order to compute $\rho_D$ including terms of order $\theta^2$. We also use the previously obtained ``simple properties'' for $\rho_D$ whenever it is multiplied by operators and factors of order $\theta$ or higher, like e.g.~$\theta^2 \mathbf{a} \rho_D \mathbf{a}^\dagger = \theta^2 \mathbf{a} ( e^{i\pi{\mathbf{N}}} \, \rho_D \, e^{i\pi{\mathbf{N}}}) \mathbf{a}^\dagger \, + O(\theta^3)$. This yields the equation:
\begin{eqnarray*}
\rho_D & = & e^{i\pi{\mathbf{N}}}\, \rho_D\, e^{i\pi{\mathbf{N}}} + \frac{\theta \sin 2 \phi}{1-\sin2\phi}[\rho_D,\,(\mathbf{a}-\mathbf{a}^\dagger)] \\
& & - 2 \theta^2 \cos^2 2\phi \sin 2\phi \mathcal{D}(\tfrac{\mathbf{a}+\mathbf{a}^\dagger}{2}) \rho_D \\
& & + \frac{\theta^2}{2}\, \mathcal{D}(\mathbf{a}) \rho_D + \frac{\theta^2}{2}\, \mathcal{D}(\mathbf{a}^\dagger) \rho_D \\
& & + \frac{\theta^2}{2} \cos^2 2\phi \left( \rho_D + \mathbf{a}\rho_D \mathbf{a}^\dagger - \mathbf{a}^\dagger\rho_D \mathbf{a} \right) \, .
\end{eqnarray*}
With this expression we compute the values of $\langle \mathbf{X}_{0} \rangle$ (confirming no contribution at order $\theta^2$), of $\langle \mathbf{X}_{\pi/2} \rangle$ as a check (it should be no contribution at any order), and finally of $\,\langle (\mathbf{X}_{\pi/2})^2 \rangle\,$ to get the result we're interested in. In particular, we use that $e^{i\pi{\mathbf{N}}} \mathbf{X}^2 e^{i\pi{\mathbf{N}}} = \mathbf{X}^2$ to treat the first term on the right hand side.


\end{document}